\documentclass[a4paper,UKenglish,cleveref,autoref,thm-restate,numberwithinsect]{lipics-v2021}

\usepackage{amssymb}
\usepackage{amsmath}
\usepackage{latexsym}
\usepackage{url}
\usepackage{proof}
\usepackage{tikz}
\usetikzlibrary{positioning,arrows}
\usetikzlibrary{decorations.pathreplacing}
\usetikzlibrary{decorations.pathmorphing}
\usetikzlibrary{patterns}
\usepackage{ifthen}
\usepackage[normalem]{ulem}

\usepackage{scalerel}
\usepackage{proof}
\usepackage{ifthen}
\usepackage{mathtools}

\usepackage{tikz}
\usetikzlibrary{positioning,arrows,calc}
\usepackage{amsmath,amssymb}

\newcommand{\m}[1]{\mathsf{#1}}
\newcommand{\mr}[1]{\mathrel{#1}}
\newcommand{\seq}[2][n]{{#2_1},\dots,{#2_{#1}}}
\newcommand{\sig}[2][n]{{#2_1}\times\cdots\times{#2_{#1}}}
\newcommand{\h}[1][.3]{\hspace{#1mm}}
\newcommand{\SET}[1]{\{\h#1\h\}}
\newcommand{\inter}[1]{[\![{#1}]\!]}

\newcommand{\xB}{\mathcal{B}}

\newcommand{\xE}{\mathcal{E}}
\newcommand{\xF}{\mathcal{F}}

\newcommand{\xI}{\mathcal{I}}
\newcommand{\xJ}{\mathcal{J}}
\newcommand{\xM}{\mathcal{M}}
\newcommand{\xR}{\mathcal{R}}
\newcommand{\xS}{\mathcal{S}}
\newcommand{\xT}{\mathcal{T}}

\newcommand{\xV}{\mathcal{V}}
\newcommand{\xFTe}{\xF_{\m{te}}}
\newcommand{\xFTh}{\xF_{\m{th}}}
\newcommand{\xSTe}{\xS_{\m{te}}}
\newcommand{\xSTh}{\xS_{\m{th}}}

\newcommand{\Val}{\xV\m{al}}
\newcommand{\Var}{\xV\m{ar}}
\newcommand{\FVar}{\xF\xV\m{ar}}
\newcommand{\BVar}{\xB\xV\m{ar}}

\newcommand{\ExVar}{\mathcal{E}\m{x}\Var}
\newcommand{\Dom}{\mathcal{D}\m{om}}
\newcommand{\Ran}{\mathcal{R}\m{an}}

\newcommand{\VDom}{\mathcal{V}\Dom}
\newcommand{\Pos}{\mathcal{P}\m{os}}
\newcommand{\VPos}{\Pos_\xV}
\newcommand{\FPos}{\Pos_\xF}

\newcommand{\sort}[1]{\m{#1}}
\newcommand{\Bool}{\sort{Bool}}
\newcommand{\Int}{\sort{Int}}

\newcommand{\CO}[1]{[\h#1\h]} 
\newcommand{\ECO}[2]{\exists #1.\ #2}

\newcommand{\CTerm}[4]{%
\ifthenelse{\equal{#1}{} \and \equal{#2}{} \and \equal{#3}{} \and \equal{#4}{}}%
{\Pi X.\ s~\CO{\ECO{\vec{x}}{\varphi}}}%
{%
\ifthenelse{\equal{#1}{}}{}{\Pi #1.\ }%
#2%
\ifthenelse{\equal{#4}{}}%
{}%
{\ifthenelse{\equal{#3}{}}{~\CO{#4}}{~\CO{\ECO{#3}{#4}}}}%
}%
}
\newcommand{\AllInst}[1]{\langle\!\langle #1 \rangle\!\rangle}

\newcommand{\CEqn}[4]{%
\ifthenelse{\equal{#1}{}}{%
\ifthenelse{\equal{#4}{}}{#2 \approx #3}{#2 \approx #3~\CO{#4}}}{%
\ifthenelse{\equal{#4}{}}{\Pi #1.\, #2 \approx #3}{\Pi #1.\, #2 \approx #3~\CO{#4}}%
}}

\newcommand{\CRu}[4]{%
\ifthenelse{\equal{#1}{} \and \equal{#2}{} \and \equal{#3}{} \and \equal{#4}{}}%
{\Pi X.\ \ell \R r~\CO{\varphi}}%
{%
\ifthenelse{\equal{#1}{}}{}{\Pi #1.\ }%
#2 \R #3%
\ifthenelse{\equal{#4}{}}%
{}%
{~\CO{#4}%
}%
}%
}

\newcommand{\R}{\rightarrow}

\renewcommand{\L}{\leftarrow}

\newcommand{\Lb}[1][]{\mr{\vphantom{\R}_{#1}{\L}}}

\newcommand{\Rbase}[1][]{\R_{\mathsf{base}}}

\newcommand{\Ca}[1][]{\xleftrightarrow{#1}}

\newcommand{\Cb}[1][]{\Ca[]_{#1}}

\newcommand{\Cru}[1][\xE]{\ifthenelse{\equal{#1}{}}{\Cb[\m{rule}]}{\Cb[\m{rule},#1]}}
\newcommand{\Cbase}[1][\xE]{\ifthenelse{\equal{#1}{}}{\Cb[\m{base}]}{\Cb[\m{base},#1]}}

\renewcommand{\ge}{\geqslant}
\renewcommand{\le}{\leqslant}

\newcommand\subsetsim{\mathrel{\substack{
  \textstyle\subset\\[-0.2ex]\textstyle\sim}}}
\newcommand\supsetsim{\mathrel{\substack{
  \textstyle\supset\\[-0.2ex]\textstyle\sim}}}

\newcommand{\Bfnum}[1]{\textcolor{darkgray}{\rm\sffamily\bfseries #1}}

\newcommand{\crest}{\textsf{crest}}

\newcommand{\pvec}[1]{\vec{#1}\mkern2mu\vphantom{#1}}

\title{
Partial Rewriting and Value Interpretation of Logically Constrained Terms
}

\author{Takahito Aoto}{Niigata University, Japan}{aoto@ie.niigata-u.ac.jp}
{https://orcid.org/0000-0003-0027-0759}
{JSPS KAKENHI Grant Number JP24K14817.}

\author{Naoki Nishida}{
Nagoya University, Japan}{nishida@i.nagoya-u.ac.jp}{https://orcid.org/0000-0001-8697-4970}{JSPS KAKENHI Grant Number JP24K02900.}

\author{Jonas Sch\"opf}{Data Lab Hell, Austria}
{jonas.schoepf.pub@datalabhell.at}{https://orcid.org/0000-0001-5908-8519}{}

\authorrunning{T. Aoto, N. Nishida, and J. Sch\"opf}

\Copyright{Takahito Aoto, Naoki Nishida, and Jonas Sch\"opf}



\category{} 

\funding{}

\acknowledgements{}

\nolinenumbers 

\EventEditors{}
\EventNoEds{2}
\EventLongTitle{}
\EventShortTitle{}
\EventAcronym{}
\EventYear{}
\EventDate{}
\EventLocation{}
\EventLogo{}
\SeriesVolume{}
\ArticleNo{}

\begin{document}
\maketitle              
\begin{abstract}
Logically constrained term rewrite systems (LCTRSs) are a rewriting
formalism that naturally supports built-in data structures, including
integers and bit-vectors. The recent framework of existentially constrained
terms and most general constrained rewriting on them (Takahata et al.,
2025) has many advantages over the original approach of rewriting constrained
terms. In this paper, we introduce partial constrained rewriting, a variant of rewriting
existentially constrained terms whose underlying idea has already appeared
implicitly in previous analyses and applications of LCTRSs. We examine the
differences between these two notions of constrained
rewriting. First, we establish a direct correspondence between them,
leveraging subsumption and equivalence of constrained terms where
appropriate. Then we give characterizations of each of them, using the
interpretation of existentially constrained terms by instantiation. We further introduce
the novel notion of value interpretation, that highlights subtle
differences between partial and most general rewriting.
\end{abstract}

\bibliographystyle{plainurl}

\section{Introduction}

Logically constrained term rewrite systems (LCTRSs) are a formalism that extends
term rewrite systems (TRSs). Many real-world computational problems
involve built-in data structures, such as integers and bit-vectors,
that often limit the applicability of techniques developed for
standard TRSs. By equipping rewrite rules with constraints interpreted
over an arbitrary model, the LCTRS formalism naturally supports such
built-in data structures. A constrained rewrite rule can be applied
only when its instantiated constraint is valid in the
underlying model; consequently, the framework is sufficiently
general to accommodate arbitrary built-in data structures in theory.
Moreover, by incorporating support from SMT-solvers and
leveraging the theories in common SMT libraries, LCTRS tools
can effectively address rewriting problems involving these data
structures.

Many topics from the term rewriting literature have already been
studied within the LCTRS formalism, including such as confluence
analysis~\cite{SM23,SMM24}, (non-)termination analysis~\cite{Kop13,NW18},
completion~\cite{WM18}, rewriting induction~\cite{KN14aplas,FKN17tocl},
algebraic semantics~\cite{ANS24}, and complexity analysis~\cite{WM21}.
Nevertheless, the framework also exhibits aspects that are not
covered by traditional term rewriting. One of such aspect is the
\emph{rewriting of} constrained terms. A constrained term consists of
a term together with a constraint that restricts the ways in
which the term may be instantiated. For example, the constrained term
$\m{f}(x)~\CO{x > 2}$ (in LCTRS notation) can be intuitively viewed as a
set of terms $\SET{ \m{f}(x) \mid x > 2 }$. In the analysis of LCTRSs,
constrained rewrite rules must be handled, and consequently rewriting
constrained terms naturally arises. Indeed, constrained rewriting is an integral
part of many different analysis techniques. For instance, in confluence
analysis, finding a specific joining sequence for constrained critical
pairs often requires reasoning about two terms under a shared
constraint, which leads to working on constrained terms. Similarly, in
rewriting induction, rewriting of constrained terms is employed in several
of its inference steps.

Rewriting constrained terms was already introduced in early work
on LCTRSs~\cite{KN13frocos}. However, the original formulation of
rewriting constrained terms is difficult to deal with, since rewrite
steps are used modulo equivalence transformations of constrained
terms, and the results of rewriting a constrained term remain highly
flexible even when both the rewrite rule and the position are
fixed. These difficulties motivated us 
to revisit the formalism of
constrained rewriting for LCTRSs in
~\cite{TSNA25-PPDP}. In
that work, we have introduced the notions of existentially constrained terms
and most general constrained rewriting for left-linear LCTRSs. We have shown the
uniqueness of reducts in this new setting, as well as commutativity
between rewrite steps and equivalence, under the restriction to
left-value-free rewrite rules\footnote{However, any rewrite rule is simulated by
a left-value-free rewrite rule, and thus this is not a real
restriction.}. Our formalism of most general constrained rewriting extracts the
so-called ``most general'' part of the corresponding original formalism.
As a result, it offers several advantages over the original approach
of rewriting constrained terms. 

In this paper, we introduce another notion of rewriting on existentially
constrained terms, called \emph{partial constrained rewriting}. The difference
between the two notions can be explained informally as follows. On the
one hand, most general constrained rewriting applies a rule to
\emph{all} instantiations of a constrained term. On the other hand,
partial constrained rewriting applies a rule to only \emph{some}
instantiations of a constrained term (see \Cref{fig:intuition for two
kinds of rewriting}).

\begin{example}
\label{exp:strong rewrite step}
Let us consider the following LCTRS over the integer theory:
\[
\mathcal{R} = 
\left\{
\begin{array}{l@{\>}c@{\>}ll}
\mathsf{sum}(x) & \to&  0 &  [0 \ge x ]\\
\mathsf{sum}(x) &\to&  x + \mathsf{sum}(x + -1) &   [ 0 < x ]
\end{array}
\right
\}
\]
The following reduction obtains a normal form by most general constrained rewriting:
\[
\begin{array}{@{}l@{\>}c@{\>}l@{\>}c@{}}
\mathsf{sum}(x)~[1 \le x \land x \le 5] 
& \to_{\xR} &  x + \mathsf{sum}(x + -1)~[1 \le x \land x \le 5]\\
& \to_{\xR} &  x + \mathsf{sum}(y)~[1 \le x \land x \le 5 \land y = x - 1]
& \not\to_{\xR} 
\end{array}
\]
On the other hand,
we have 
$\mathsf{sum}(1) \to_{\xR}^*  1$,
$\mathsf{sum}(2) \to_{\xR}^*  3$,
$\mathsf{sum}(3) \to_{\xR}^*  6$,
$\mathsf{sum}(4) \to_{\xR}^*  10$, and
$\mathsf{sum}(5) \to_{\xR}^*  15$.
However, most general rewriting is not suitable for this computation, 
as there is no single rule that can be applied
to the constrained term $\mathsf{sum}(y)~[1 \le x \land x \le 5 \land y = x - 1]$.
In contrast, partial constrained rewriting allows
the rewriting process to proceed by
\[
\begin{array}{@{}l@{\>}c@{\>}l@{\>}c@{\>}l@{}}
\cdots & \to_{\xR} &  x + \mathsf{sum}(y)~[1 \le x \land x \le 5 \land y = x - 1]
& \leadsto_{\xR} & 1 + 0\\
\end{array}
\]
because the first rule is applicable 
as the constraint $1 \le x \land x \le 5 \land y = x - 1$ is valid
with $x = 1$ and $y = 0$.
Similarly, partial constrained rewriting allows 
us to continue using the second rewrite rule as
\[
\begin{array}{@{}l@{\>}c@{\>}l@{}}
\cdots & \to_{\xR} &  1 + \mathsf{sum}(y)~[1 \le x \land x \le 5 \land y = x - 1]\\
       &\leadsto_{\xR} &  1 + (1 + \mathsf{sum}(z))~[1 \le x \land x \le 5 \land y = x - 1 \land z = y - 1]
\end{array}
\]
Thus, by partial rewriting we obtain
$\SET{ t \in \mathsf{NF} \mid \mathsf{sum}(x)~[1 \le x \land x \le 5] \leadsto^* t }
= \SET{ 1, 3, 6, 10, 15 }$.
\end{example}
The motivation for introducing this variant of rewriting on existentially
constrained terms is twofold: not only is the idea itself natural, but it has
also appeared implicitly in previous analyses and applications of LCTRSs,
although the two rewriting formalisms may need to be applied separately
depending on the case.

In this paper, we study the  differences between most general rewriting and
partial rewriting from several perspectives. First, we establish a direct
correspondence between the two formalisms, showing that partial
constrained rewrite steps are included in most general rewrite steps, but not
vice versa. We further demonstrate that most general rewrite steps can
be simulated by partial rewriting through subsumption and equivalence of
constrained terms.
We also provide characterizations of each type of rewriting
constrained terms in terms of interpretations of constrained terms and
rewrite rules as sets of terms and rewrite rules, respectively.
The idea of interpreting constrained terms as sets of terms
has appeared, for example, in~\cite[Section~2.3]{FKN17tocl} to reduce
certain problems on LCTRSs to problems in TRSs. However, these
characterizations alone are insufficient to distinguish partial
constrained reductions from most general constrained reductions.
This motivates the introduction of a novel notion of \emph{value
interpretation}. Like the standard interpretation, the value interpretation
reflects the equivalence of constrained terms; however, unlike the
standard interpretation, it does not reflect subsumption. Using the
value interpretation, we provide more precise characterizations of
the two notions of rewriting constrained terms, highlighting
subtle differences between partial and most general rewriting.

\begin{figure}
\begin{tikzpicture}[scale=.8,thick]
\node at (0,5) {Two kinds of rewrite steps by the rule};
\draw[dashed]  (4,4.5) rectangle (4.5,5.5);
\draw[->,thick] (4.7,5) -- (5.2,5);
\draw[pattern=north east lines, pattern color=gray](5.5,4.5) rectangle (6,5.5);
\draw[->,thick] (2.2,2.5) -- (3,2.5);
\draw[dashed] (0,1) rectangle (2,4);
\draw (1,2.5) circle (.7 and 1.2);
\draw [pattern=north east lines, pattern color=gray] (4,2.5) circle (.7 and 1.2);
\draw[dashed] (7,1) rectangle (9,4);
\node at (2,-1) {Most General Rewrite Step};
\node at (9,-1) {Partial Rewrite Step};
\draw (8,1) circle (.7 and 1.2);
\draw[->,thick,decorate,decoration=snake] (9.2,1)-- (10,1);
\begin{scope}
\clip (10,1) rectangle (13,3); 
\draw [pattern=north east lines, pattern color=gray] (11,1) circle (.7 and 1.2);
\end{scope}
\draw (10.3,1) -- (11.7,1);
\end{tikzpicture}
\caption{Most General and Partial Rewrite Steps}
\label{fig:intuition for two kinds of rewriting}
\end{figure}
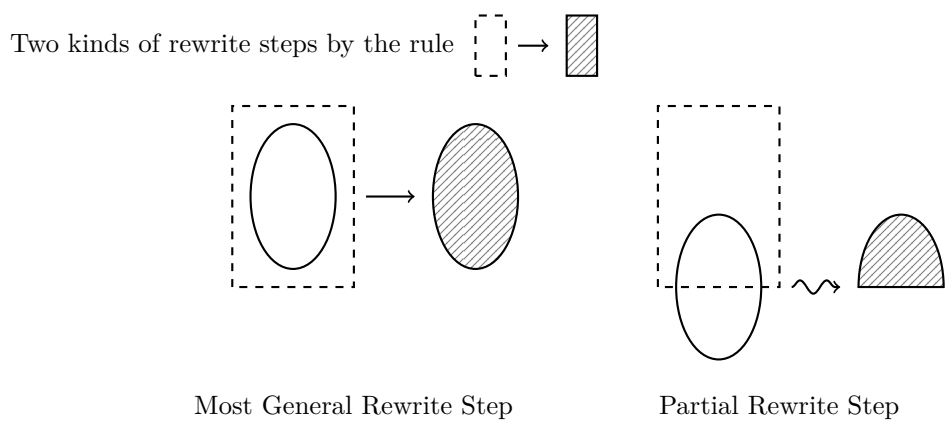

The remainder of the paper is organized as follows.
\Cref{sec:preliminaries} introduces some preliminaries.
In \Cref{sec:partial-reduction}, we introduce partial constrained rewriting
and establish direct correspondences
between partial constrained rewriting and most general constrained rewriting.
\Cref{sec:Interpretations of Weak and Most General Reductions} explores properties of 
interpretations of existentially constrained terms,
introduces the interpretation of constrained rewrite rules,
and characterizes both partial and most general rewrite steps
based on these interpretations, while also highlighting
the limitations of these.
In \Cref{sec:value-interpretation}, we introduce the value interpretation and discuss
its properties. 
\Cref{sec:characterization} provides characterizations of partial 
and most general rewrite steps using value interpretations.
Related work is discussed in \Cref{sec:related-work},
and \Cref{sec:conclusion} concludes the paper.
All omitted proofs of the claims 
are provided in the appendix.

\section{Preliminaries}
\label{sec:preliminaries}

In this section, we briefly recall the basic notions of LCTRSs~\cite{ANS24,KN13frocos,SM23,SMM24,TSNA25LOPSTR}
and fix additional notations used throughout this paper.
Familiarity with the basic notions of term rewriting is assumed (see for example~\cite{BN98}).

We employ a sorted signature $\langle \xS,\xF \rangle$
consisting of a set $\xS$ of sorts 
and a set $\xF$ of function symbols,
where each $f \in \xF$ is attached with its sort declaration 
$f\colon \sig{\tau} \to \tau_0$  ($\tau_0,\ldots,\tau_n \in \xS$).
The set of sorted variables is denoted by $\xV$.
For $X \subseteq \xV$ and $F \subseteq \xF$,
the set of terms over the function symbols in $F$ and 
variables in $X$ is denoted by $\xT(F,X)$.
We write $t^\tau$ to denote that a term $t$ has sort $\tau$.
The set of variables appearing in a term $t$ is denoted by $\Var(t)$.
A term $t$ is said to be \emph{linear} when any variable occurs at most once in $t$.
We denote a sequence of variables by $\vec{x}$, and use the same notation for sequences of function symbols, terms, or similar objects.
The set of variables occurring in $\vec{x}$ is denoted by $\SET{\vec{x}}$.
Similarly, the set of sequences of elements of a set $T$ is denoted by $T^*$.

The set of positions in a term $t$ is denoted by $\Pos(t)$.
The symbol and subterm occurring at a position $p \in \Pos(t)$ is denoted
by $t(p)$ and $t|_p$, respectively. 
A term obtained from $t$ by replacing subterms at positions $\seq{p}$ 
by the terms $\seq{t}$, having the same sort as $t|_{p_1},\ldots,t|_{p_n}$,
is written as $t[\seq{t}]_{\seq{p}}$ or just $t[\seq{t}]$ when no confusions arises.
If we consider an expression obtained by
replacing those subterms $t|_{p_1},\ldots,t|_{p_n}$ in $t$ by holes of the same sorts,
it is called a multihole context and denoted by $t[~]_{\seq{p}}$.
A sort-preserving function $\sigma$ from $\xV$ to $\xT(\xF,\xV)$ 
with a finite domain $\Dom(\sigma) = \SET{x \in \xV \mid \sigma(x) \neq x}$ 
is called a substitution. It is naturally extended to a homomorphism
$\sigma\colon \xT(\xF,\xV) \to \xT(\xF,\xV)$, which we also denote by $\sigma$.
For a set of terms $U$, we write $\sigma(U) = \SET{\sigma(t) \mid t \in U}$.
As usual, $\sigma(t)$ can also be also written as $t\sigma$.
We write $\SET{ x_1\mapsto t_1,\ldots, x_n \mapsto} t_n $
for the substitution $\sigma$
such that $\sigma(x_i) = t_i$ ($1 \le i \le n)$
and $\Dom(\sigma) \subseteq \SET{ x_1,\ldots,x_n }$, 
which may be abbreviated as $\SET{ \vec{x} \mapsto \vec{t} }$.

We assume that the set $\xS$ of sorts and the set $\xF$ of function symbols 
are partitioned into two disjoint sets, respectively: 
$\xS = \xSTh \uplus \xSTe$ and $\xF = \xFTh \uplus \xFTe$,
where each $f\colon \sig{\tau} \to \tau_0 \in \xFTh$ satisfies 
$\tau_i \in \xSTh$ for all $0 \leqslant i \leqslant n$. 
Elements of $\xSTh$ ($\xFTh$) and $\xSTe$ ($\xFTe$)
are called theory sorts (symbols) and term sorts (symbols).
Moreover, we consider a model $\xM = \langle \xI,\xJ \rangle$ 
over the sorted signature $\langle \xSTh,\xFTh \rangle$.
Here, $\xI$ is the interpretation of theory sorts
and assigns any sort $\tau \in \xSTh$ a non-empty set $\xI(\tau)$.
Theory symbols are interpreted by $\xJ$,
which assigns any theory symbol $f\colon\sig{\tau} \to \tau_0 \in \xFTh$ 
to a function $\xJ(f)\colon \xI(\tau_1) \times \cdots \times \xI(\tau_n) \to \xI(\tau_0)$.
A valuation $\rho$ on the model $\xM = \langle \xI, \xJ \rangle$
is a mapping that assigns any $x^\tau \in \xV$ with $\rho(x) \in \xI(\tau)$.
The interpretation of a term $t$ in the model $\xM$ over the valuation $\rho$
is denoted by $\inter{t}_{\xM,\rho}$.

We further assume that 
any element $a \in \xI(\tau)$ appears as a constant $a^\tau$ in $\xFTh$,
for every sort $\tau \in \xSTh$.
These constants are called \emph{values} and the set of all values is denoted by $\Val$.
The set of values contained in a term $t$ is denoted by $\Val(t)$.
A term $t$ is \emph{value-free} if $\Val(t) = \varnothing$.
For a substitution $\gamma$ we define $\VDom(\gamma) = \SET{x \in \xV \mid \gamma(x) \in \Val}$.

We also assume a special sort $\Bool \in \xSTh$ for Booleans
having the standard interpretation
$\xI(\Bool) = \mathbb{B} = \SET{\mathsf{true}, \mathsf{false}}$,
and the existence of necessary theory symbols 
such as $\neg$, ${\land}$, ${\Rightarrow}$, ${\Leftrightarrow}$, ${=_\tau}$, \ldots
with their standard interpretations.
The terms in $\xT(\xFTh,\xV)$ having sort $\Bool$ are called
\emph{logical constraints} (or simply \emph{constraints}).
For a constraint $\varphi$ and a valuation $\rho$ over the model $\xM$, 
we write $\vDash_{\xM,\rho} \varphi$ if $\inter{\varphi}_{\xM,\rho} = \mathsf{true}$, 
and $\vDash_{\xM} \varphi$ if 
$\vDash_{\xM,\rho} \varphi$ for all valuations $\rho$.
In the latter case the constraint $\varphi$ is said to be \emph{valid}.
A constraint $\varphi$ is said to be \emph{satisfiable}
if $\vDash_{\xM,\rho} \varphi$ for some valuation $\rho$.
For $X \subseteq \xV$, a substitution $\gamma$ is said to be \emph{$X$-valued}
if $\gamma(X) \subseteq \Val$.
We write $\gamma \vDash_\xM \varphi$ (and say \emph{$\gamma$ respects $\varphi$})
if the substitution $\gamma$ is $\Var(\varphi)$-valued 
and $\vDash_{\xM} \varphi\gamma$.
If no confusion arises then we drop the subscript $\xM$.

\begin{example}
\label{exp:sorted-signature-for-leading-example}
For the LCTRS in \Cref{exp:strong rewrite step}
we have $\xSTh = \SET{ \Bool, \Int }$, $\xSTe = \varnothing$,
$\xFTh = \mathbb{Z} \cup \mathbb{B} \cup 
\SET{ +: \Int \times \Int \to \Int, 
{\ge}:\Int \times \Int \to \Bool,
{<}:\Int \times \Int \to \Bool, \ldots }$,
and $\xFTe = \SET{ \mathsf{sum}: \Int \times \Int \to \Int }$,
where $\Val = \mathbb{Z} \cup \mathbb{B}$.
We have $\xI(\Bool) = \mathbb{B}$, $\xI(\Int) = \mathbb{Z}$,
$\xJ(+): \mathbb{Z} \times \mathbb{Z} \to \mathbb{Z}$
is the addition function on $\mathbb{Z}$, and so on.
If we take a valuation $\rho$ such that $\rho(x) = \rho(y) = 1$,
then we have $\inter{x + y}_{\xM,\rho} = 2$
and $\inter{x + y \ge 2}_{\xM,\rho} = \mathsf{true}$.
Thus, $\vDash_{\xM,\rho} x + y \ge 2$, as $x + y \ge 2$ is a constraint.
The constraint $x + y \ge 2$ is satisfiable but not valid.
For example, 
$x =_\mathbb{Z} x$, $(x \ge 0) \lor (x < 0)$ and 
$(0 < x ) \Rightarrow (x \ge 0)$ are valid constraints, hence we obtain
$\vDash_{\xM} (x \ge 0) \lor (x < 0)$.
Note that values are terms, therefore
$\sigma = \SET{ x \mapsto 1, y \mapsto 1, z  \mapsto \mathsf{sum}(w) }$
is a substitution. We have $(x + z)\sigma = 1 + \mathsf{sum}(w)$.
As $\VDom(\sigma) = \SET{ x,y  }$,
the substitution $\sigma$ is $\SET{x,y}$-valued which implies
$\sigma \vDash_{\xM} x + y \ge 2$.
\end{example}

An \emph{existential constraint} is a pair $\langle \vec{x},
\varphi \rangle$ of a sequence of variables $\vec{x}$ and a constraint
$\varphi$, written as $\ECO{\vec{x}}{\varphi}$, 
such that $\SET{\vec{x}} \subseteq \Var(\varphi)$.
We define the
sets of \emph{free} and \emph{bound} variables of $\ECO{\vec{x}}{\varphi}$
as follows:
$\FVar(\ECO{\vec{x}}{\varphi}) = \Var(\varphi) \setminus \SET{\vec{x}}$
and
$\BVar(\ECO{\vec{x}}{\varphi}) = \SET{\vec{x}}$.
We define the interpretation of existential constraints as follows:
We have $\vDash_{\xM,\rho} \ECO{\vec{x}}{\varphi}$
if there exists a valuation $\rho' = \rho[\vec{x} \mapsto \vec{w}]$,
i.e.\ $\rho'(x) = \rho(x)$ for any $x \in \xV \setminus \SET{\vec{x}}$
and $\rho'(\vec{x}) = \vec{w}$,
such that $\xM \vDash_{\xM,\rho'} \varphi$.
An \emph{existentially constrained term} is a 
triple $\langle X, s, \ECO{\vec{x}}{\varphi} \rangle$, written as $\CTerm{X}{s}{\vec{x}}{\varphi}$, 
of a set $X$ of variables,
a term $s$, and an existentially constraint $\ECO{\vec{x}}{\varphi}$,
such that
$\FVar(\ECO{\vec{x}}{\varphi}) \subseteq X \subseteq \Var(s)$,
and $\BVar(\ECO{\vec{x}}{\varphi}) \cap \Var(s) = \varnothing$.
Variables in $X$ are called \emph{logical variables} (of $\CTerm{X}{s}{\vec{x}}{\varphi}$).
An existentially constrained term $\CTerm{X}{s}{\vec{x}}{\varphi}$ is said to be \emph{satisfiable} if 
$\ECO{\vec{x}}{\varphi}$ is satisfiable.
\begin{example}
\label{exp:existential-constraint}
Let us continue with the LCTRS and the sorted signature in \Cref{exp:strong rewrite step,exp:sorted-signature-for-leading-example}.
Then $\ECO{x}{x + y \ge 2}$ is an example of an existential constraint
for which we have
$\FVar(\ECO{x}{x + y \ge 2}) = \SET{ y }$
and
$\BVar(\ECO{x}{x + y \ge 2}) = \SET{ x }$.
Also $\ECO{x,y}{x + y \ge 2}$ and $x + y \ge 2$ are existential constraints.
Note that for the latter we drop the sequence of variables as it is empty.
On the other hand, 
we exclude $\ECO{z}{x + y \ge 2}$ from the set of existential constraints,
as $z \notin \Var({x + y \ge 2})$.
Consider $\rho$ from the example above
that yields $\vDash_{\xM,\rho} \ECO{w}{x + y \ge w}$.
Existential constraints can be used to construct the following 
two examples of existentially constrained terms 
$\CTerm{\SET{x,y }}{\mathsf{sum}(x,y)}{w}{x + y \ge w}$
and 
$\CTerm{\SET{x,y }}{\mathsf{sum}(x,y)}{}{x \ge y}$.
The $\Pi$-notation is designed to allow value-only instantiations for
variables that do not appear in the constraint part.
For example, the constrained term
$\CTerm{\SET{x }}{\mathsf{sum}(x,y)}{}{\mathsf{true}}$
allows the instantiations $n \in \mathbb{Z}$
and $t \in \xT(\xF,\xV)$ in the term $\mathsf{sum}(n,t)$.
\end{example}

An existentially constrained term $\CTerm{X}{s}{\vec{x}}{\varphi}$ is said to be
\emph{subsumed by} an existentially constrained term
$\CTerm{Y}{t}{\vec{y}}{\psi}$, denoted by $\CTerm{X}{s}{\vec{x}}{\varphi}
\subsetsim \CTerm{Y}{t}{\vec{y}}{\psi}$, 
if for all $X$-valued substitutions
$\sigma$ with 
$\sigma \vDash_\xM \ECO{\vec{x}}{\varphi}$ 
there exists a
$Y$-valued substitution $\gamma$ with 
$\gamma \vDash_\xM \ECO{\vec{y}}{\psi}$
such that $s\sigma = t\gamma$.
Two existentially constrained terms $\CTerm{X}{s}{\vec{x}}{\varphi}$ 
and $\CTerm{Y}{t}{\vec{y}}{\psi}$ are said to be \emph{equivalent},
denoted by $\CTerm{X}{s}{\vec{x}}{\varphi}
\sim \CTerm{Y}{t}{\vec{y}}{\psi}$
if $\CTerm{X}{s}{\vec{x}}{\varphi} \subsetsim \CTerm{Y}{t}{\vec{y}}{\psi}$
and $\CTerm{Y}{t}{\vec{y}}{\psi} \subsetsim \CTerm{X}{s}{\vec{x}}{\varphi}$.
\begin{example}
Let us continue with \Cref{exp:existential-constraint}.
We obtain that
$\CTerm{\SET{x,y }}{\mathsf{sum}(x,y)}{}{x \ge y}
\subsetsim
\CTerm{\SET{x',y' }}{\mathsf{sum}(x',y')}{w}{w \ge -1 \land x' = y' + w}$
and
$\CTerm{\SET{x,y }}{\mathsf{sum}(x,y)}{}{\mathsf{true}}
\subsetsim 
\CTerm{\SET{x' }}{\mathsf{sum}(x',y')}{}{\mathsf{true}}$,
but both do not satisfy equivalence.
We have, for example, 
$\CTerm{\SET{x,y }}{\mathsf{sum}(x,y)}{}{x = 1 \land y > x}
\sim \CTerm{\SET{y' }}{\mathsf{sum}(1,y')}{}{y' \ge 2}$.
For various characterizations of equivalence, see~\cite{TSNA25LOPSTR}.
\end{example}



A \emph{constrained rewrite rule} is a quadruple
consisting of a set $Z$ of theory variables,
terms $\ell, r$ of a same sort and a constraint $\pi$
such that $\Var(\pi) \cup (\Var(r)\setminus \Var(\ell)) \subseteq Z$,
which is written as $\CRu{Z}{\ell}{r}{\pi}$.
A constrained rule $\CRu{Z}{\ell}{r}{\pi}$
is said to be \emph{left-linear} (\emph{left-value-free})
if $\ell$ is linear (value-free).
The original definition of constrained rewriting 
in~\cite{KN13frocos,FKN17tocl} is defined modulo equivalence $\sim$ of constrained
terms, which makes analyses more complicated and its implementation mostly infeasible.
However, most general constrained rewriting on existential constrained terms~\cite{TSNA25-PPDP}
has several advantages compared to the original definition of constrained rewriting,
even though it is presently formulated only for left-linear constrained rewrite rules.
Since this paper follows the line of our previous work~\cite{TSNA25-PPDP},
we again restrict our attention to left-linear constrained rewrite rules.

Let $\CTerm{X}{s}{\vec{x}}{\varphi}$ be a satisfiable existentially constrained term. 
Suppose
that $\rho\colon \CRu{Z}{\ell}{r}{\pi}$ is a left-linear constrained rule
satisfying $\Var(\rho) \cap \Var(s,\varphi) = \varnothing$. We say
\emph{$\CTerm{X}{s}{\vec{x}}{\varphi}$ has a $\rho$-redex at position $p \in
\Pos(s)$ using substitution $\gamma$} if 
\Bfnum{1.}
$\Dom(\gamma) = \Var(\ell)$,
\Bfnum{2.}
$s|_p = \ell\gamma$,
\Bfnum{3.}
$\gamma(x) \in \Val \cup X$ for all $x \in \Var(\ell) \cap Z$, and
\Bfnum{4.}
$\vDash_{\xM} (\ECO{\vec{x}}{\varphi}) \Rightarrow (\ECO{\vec{z}}{\pi\gamma})$,
where $\SET{\vec{z}} = \Var(\pi) \setminus \Var(\ell)$.
Suppose $\CTerm{X}{s}{\vec{x}}{\varphi}$ has a $\rho$-redex 
at position $p \in \Pos(s)$ using substitution $\gamma$.
Then we have a \emph{most general rewrite step}
$\CTerm{X}{s}{\vec{x}}{\varphi} \R_\rho \CTerm{Y}{t}{\vec{y}}{\psi}$
(or written as $\CTerm{X}{s}{\vec{x}}{\varphi} \R^p_{\rho,\gamma} \CTerm{Y}{t}{\vec{y}}{\psi}$
with explicit $p$ and $\gamma$)
where
\Bfnum{1.}
        $t = s[r\gamma]$,
\Bfnum{2.}
        $\psi = \varphi \land \pi\gamma$,
\Bfnum{3.}
        $\SET{\vec{y}} = \Var(\psi) \setminus \Var(t)$, and
\Bfnum{4.}
        $Y = \ExVar(\rho) \cup (X \cap \Var(t))$.
An LCTRS is a set of constrained rewrite rules
that includes the set of calculation rules
$\xR_\textsf{ca} = \SET{ f(x_1,\ldots,x_n) \to y ~\CO{y = f(x_1,\ldots,x_n)} \mid f
\in \xFTh \setminus \Val }$. However, in examples we will refrain from
explicitly mentioning the calculation rules of $\xR$.
For an LCTRS $\xR$, we have 
$\CTerm{X}{s}{\vec{x}}{\varphi} \R_\xR \CTerm{Y}{t}{\vec{y}}{\psi}$
if 
$\CTerm{X}{s}{\vec{x}}{\varphi} \R_\rho \CTerm{Y}{t}{\vec{y}}{\psi}$
for some $\rho \in \xR$.
The transitive and reflexive closure
$\R_\xR^*$
of $\R_\xR$ is called the \emph{most general constrained reduction}.


\begin{example}
\label{exp:most general rewrite steps}
Following our formulation, the constrained rewrite rules in \Cref{exp:strong rewrite step}
are presented as
$\CRu{\SET{x}}{\mathsf{sum}(x)}{0}{0 \ge x}$
and
$\CRu{\SET{x}}{\mathsf{sum}(x)}{x + \mathsf{sum}(x + -1)}{x > 0}$. 
Let $s$ be the
constrained term $\CTerm{\SET{x}}{\mathsf{sum}(x)}{}{x > 2}$. In order to
perform a rewrite step starting from $s$, we first consider a renamed
variant of the second rule $\rho\colon \CRu{\SET{x'}}{\mathsf{sum}(x')}{x'
+ \mathsf{sum}(x' + -1)}{x' > 0}$.
Then $s$ admits a $\rho$-redex at position $\epsilon$ using
the substitution $\gamma = \SET{ x' \mapsto x }$, as it satisfies $\vDash_\xM x > 2 \Rightarrow x > 0$.
Hence we obtain
a most general rewrite step 
$s \to^\epsilon_{\rho,\gamma} \CTerm{\SET{x }}{x + \mathsf{sum}(x + -1)}{}{x > 2 \land x > 0}$.
In this case, the additional constraint $x > 0$ is auxiliary. 
However, this addition of a constraint is necessary in general: for example, consider a
rewrite rule $\rho': \CRu{\SET{x',y'}}{\mathsf{f}(x')}{x' + y'}{y' \ge x'}$
and a most general rewrite step
$\CTerm{\SET{x }}{\mathsf{f}(x)}{}{x > 2}
\to_{\rho'} \CTerm{\SET{ y' }}{x + y'}{}{x > 2 \land y' > x}$.
Note that the extra variable $y'$ in the rule remains uninstantiated,
since the domain of the substitution is exactly $\Var(\ell)$.
\end{example}

\section{Partial Reduction of Constrained Terms}
\label{sec:partial-reduction}

As explained in \Cref{exp:strong rewrite step},
the existing notion of rewriting logically constrained terms is
not able to capture situations in which a rewrite rule can be applied
only to a part of the instantiations (here denoted as target constrained term).
To reflect such a situation in
this section, we introduce a novel notion called \emph{partial constrained reduction}.

First, given a constrained rewrite rule $\rho$, we define the notion
of a partial $\rho$-redex before we proceed to the definition of the partial rewrite
steps.

\begin{definition}[{Partial} $\rho$-Redex]
\label{def:weak-rho-redex}
Let $\CTerm{X}{s}{\vec{x}}{\varphi}$ be a satisfiable existentially constrained term. Suppose
that $\rho\colon \CRu{Z}{\ell}{r}{\pi}$ is a left-linear constrained rule
satisfying $\Var(\rho) \cap \Var(s,\varphi) = \varnothing$. We say
\emph{$\CTerm{X}{s}{\vec{x}}{\varphi}$ has a partial $\rho$-redex at position $p \in
\Pos(s)$ using substitution $\gamma$} if all of the following statements hold:
\textup{\Bfnum{1.}}
$\Dom(\gamma) = \Var(\ell)$,
\textup{\Bfnum{2.}}
$s|_p = \ell\gamma$,
\textup{\Bfnum{3.}}
$\gamma(x) \in \Val \cup X$ for all $x \in \Var(\ell) \cap Z$, and
\textup{\Bfnum{4.}}
$(\ECO{\vec{x}}{\varphi}) \land(\ECO{\vec{z}}{\pi\gamma})$ is satisfiable in $\xM$
where $\SET{\vec{z}} = \Var(\pi) \setminus \Var(\ell)$.
\end{definition}

The difference to the definition of $\rho$-redex for a most general rewrite step
lies in item~\Bfnum{4}.
For a (non-partial) $\rho$-redex, we require the \emph{validity} of the formula
$(\ECO{\vec{x}}{\varphi}) \Rightarrow (\ECO{\vec{z}}{\pi\gamma})$,
whereas for a partial $\rho$-redex, we only require the \emph{satisfiability} of the formula
$(\ECO{\vec{x}}{\varphi}) \land (\ECO{\vec{z}}{\pi\gamma})$.
It follows that every (non-partial) $\rho$-redex is also a partial $\rho$-redex.
To be precise, this fact is stated as follows.

\begin{restatable}{lemma}{LemmaRedexIsAWeakRedex}
\label{lem:redex is a weak redex}
Let $\CTerm{X}{s}{\vec{x}}{\varphi}$ be a satisfiable existentially constrained term. Suppose
that $\rho\colon \CRu{Z}{\ell}{r}{\pi}$ is a left-linear constrained rule
satisfying $\Var(\rho) \cap \Var(s,\varphi) = \varnothing$. 
Any $\rho$-redex of $\CTerm{X}{s}{\vec{x}}{\varphi}$ at position $p \in \Pos(s)$ using a substitution $\gamma$
is a partial $\rho$-redex at $p$ using $\gamma$.
\end{restatable}

If a satisfiable constrained term contains a partial $\rho$-redex then
partial constrained reduction can be performed.

\begin{definition}[{Partial Constrained Reduction}]
    Let $\CTerm{X}{s}{\vec{x}}{\varphi}$ be a satisfiable existentially constrained term
    and suppose that $\rho\colon \CRu{Z}{\ell}{r}{\pi}$ is a left-linear constrained rule
    satisfying $\Var(\rho) \cap \Var(s,\varphi) = \varnothing$.
    Suppose $\CTerm{X}{s}{\vec{x}}{\varphi}$ has a partial $\rho$-redex 
    at position $p \in \Pos(s)$ using substitution $\gamma$.
    Then we have a \emph{partial rewrite step}
    $\CTerm{X}{s}{\vec{x}}{\varphi} \leadsto_\rho \CTerm{Y}{t}{\vec{y}}{\psi}$
    (or $\CTerm{X}{s}{\vec{x}}{\varphi} \leadsto^p_{\rho,\gamma} \CTerm{Y}{t}{\vec{y}}{\psi}$
    with explicit $p$ and $\gamma$)
    where
\textup{\Bfnum{1.}}
        $t = s[r\gamma]_p$,
\textup{\Bfnum{2.}}
        $\psi = \varphi \land \pi\gamma$,
\textup{\Bfnum{3.}}
        $\SET{\vec{y}} = \Var(\psi) \setminus \Var(t)$, and
\textup{\Bfnum{4.}}
        $Y = \ExVar(\rho) \cup (X \cap \Var(t))$.
For convenience we write 
$\CTerm{X}{s}{\vec{x}}{\varphi} \leadsto_\rho \CTerm{Y}{t}{\vec{y}}{\psi}$
if 
$\CTerm{X}{s}{\vec{x}}{\varphi} \leadsto_{\rho'} \CTerm{Y}{t}{\vec{y}}{\psi}$
for a renamed variant of $\rho$.
For an LCTRS $\xR$,
we have 
$\CTerm{X}{s}{\vec{x}}{\varphi} \leadsto_\xR \CTerm{Y}{t}{\vec{y}}{\psi}$
if 
$\CTerm{X}{s}{\vec{x}}{\varphi} \leadsto_\rho \CTerm{Y}{t}{\vec{y}}{\psi}$
for some $\rho \in \xR$.
The transitive and reflexive closure
$\stackrel{*}{\leadsto}_\xR$
of $\leadsto_\xR$ is called the \emph{partial constrained reduction} (of $\xR$).
\end{definition}

The only difference to the most general rewrite step is that the redex must be
\emph{partial}. Consequently, for the same $\rho$-redex, the reducts obtained by a
most general rewrite step and by a partial rewrite step are identical.

It follows from the definitions of $\SET{\vec{y}}$ and $Y$ that
$\CTerm{Y}{t}{\vec{y}}{\psi}$ is again an existentially constrained term; this
can be shown similarly to the case of the most general rewrite step~\cite{TSNA25-PPDP}.

\begin{example}
\label{exp: weak redex and weak rewrite steps}
Let $\rho$ be a constrained rewrite rule
$\CRu{\SET{x}}{\mathsf{sum}(x)}{0}{0 \ge x}$
and $s$ be a constrained term 
$\CTerm{\SET{y}}{1 + \mathsf{sum}(y)}{w}{1 \le w \land w \le 5 \land y = w - 1}$.
Then $s$ has a partial $\rho$-redex at position $1$ using
the substitution $\gamma = \SET{ x \mapsto y }$.
Note that
$(\exists w.\,1 \le w \land w \le 5 \land y = w - 1)
\land (0 \ge y)$ is satisfiable
as 
$\vDash_{\xM,\rho} (\exists w.\,1 \le w \land w \le 5 \land y = w - 1)
\land (0 \ge y)$ for a valuation $\rho$ with $\rho(y) = 0$.
From this we obtain
\[
s \leadsto^p_{\rho,\gamma}
\CTerm{\varnothing}{1 + 0}{w,y}{1 \le w \land w \le 5 \land y = w - 1 \land 0 \ge y}.
\]
Note that $s$ is a normal form with respect to the most general rewrite step;
this is due to the fact that 
$\vDash_\xM (\exists w.\,1 \le w \land w \le 5 \land y = w - 1) \Rightarrow (0 \ge y)$
does not hold. An obvious witness is the valuation $\rho$ with $\rho(y) = 1$.
\end{example}

We now compare partial rewrite steps with the most general rewrite steps of
existentially constrained terms. 
First, 
it follows from \Cref{lem:redex is a weak redex}
that every most general rewrite
step is also a partial rewrite step.

\begin{restatable}{theorem}{TheoremMostGeneralRewriteStepImpliesWeakRewriteStep}
\label{thm:most general rewrite step implies weak rewrite step}
Let $\CTerm{X}{s}{\vec{x}}{\varphi}$, $\CTerm{Y}{t}{\vec{y}}{\psi}$ be existentially constrained terms.
If $\CTerm{X}{s}{\vec{x}}{\varphi} \to_\rho \CTerm{Y}{t}{\vec{y}}{\psi}$ 
then 
$\CTerm{X}{s}{\vec{x}}{\varphi} \leadsto_\rho \CTerm{Y}{t}{\vec{y}}{\psi}$.
\end{restatable}

As we demonstrated in \Cref{exp: weak redex and weak rewrite steps}, partial
rewrite steps may not be most general rewrite steps. Thus, the converse 
of \Cref{thm:most general rewrite step implies weak rewrite step}
does not hold in general. However, 
one can show the following characterization of the 
partial rewrite steps.

\begin{restatable}{lemma}{LemmaFromWeakRewriteStepToMostGeneralStep}
\label{lem:from weak rewrite step to most general step}
Let $\CTerm{X}{s}{\vec{x}}{\varphi}$, $\CTerm{Y}{t}{\vec{y}}{\psi}$ be existentially constrained terms.
Suppose that $\rho\colon \CRu{Z}{\ell}{r}{\pi}$ is a left-linear constrained rule
satisfying $\Var(\rho) \cap \Var(s,\varphi) = \varnothing$.
If $\CTerm{X}{s}{\vec{x}}{\varphi} \leadsto_{\rho,\gamma} \CTerm{Y}{t}{\vec{y}}{\psi}$ 
then 
$\CTerm{X}{s}{\vec{x},\vec{z}}{\varphi \land \pi\gamma} \to_\rho 
\cdot \sim \CTerm{Y}{t}{\vec{y}}{\psi}$,
where 
$\SET{\vec{z}} = \Var(\pi) \setminus \Var(\ell)$.
\end{restatable}

From \Cref{lem:from weak rewrite step to most general step},
it follows that every partial rewrite step
is simulated by a most general rewrite step
albeit being equipped with subsumption and equivalence before and after the step.

\begin{theorem}
Let $\CTerm{X}{s}{\vec{x}}{\varphi}$, $\CTerm{Y}{t}{\vec{y}}{\psi}$ be existentially constrained terms.
If $\CTerm{X}{s}{\vec{x}}{\varphi} \leadsto_\rho \CTerm{Y}{t}{\vec{y}}{\psi}$ 
then 
$\CTerm{X}{s}{\vec{x}}{\varphi} \supsetsim \cdot \to_\rho \cdot  \sim \CTerm{Y}{t}{\vec{y}}{\psi}$.
\end{theorem}

\section{Interpretations of Partial and Most General Constrained Reductions}
\label{sec:Interpretations of Weak and Most General Reductions}

It is natural to interpret an (existentially) constrained term as the set of instances of the term
part that satisfy the constrained part. In this section, we provide
interpretations of the most general rewrite steps and partial rewrite steps with
respect to this interpretation of existentially constrained terms. This offers a more precise
characterization of each type of rewrite step, allowing the two notions of
constrained rewriting to be more clearly distinguished.

The idea that an existentially constrained term represents a set of terms
(cf.~\cite[Section~2.3]{FKN17tocl}) is formalized in our setting as follows.

\begin{definition}[Interpretation of Existentially Constrained Terms]
Let $\CTerm{X}{s}{\vec{x}}{\varphi}$ be an existentially constrained term.
A term $t$ is said to be an \emph{instance} of $\CTerm{X}{s}{\vec{x}}{\varphi}$ if 
there exists an $X$-valued substitution $\sigma$ such that 
$\sigma \vDash_\xM \ECO{\vec{x}}{\varphi}$
and $t=s\sigma$.
We denote the set of all instances of $\CTerm{X}{s}{\vec{x}}{\varphi}$ by
$\AllInst{\CTerm{X}{s}{\vec{x}}{\varphi}}$, i.e.,
$
\AllInst{\CTerm{X}{s}{\vec{x}}{\varphi}} 
=
\SET{s\sigma \mid \sigma(X) \subseteq \Val, \sigma \vDash_\xM \ECO{\vec{x}}{\varphi}}
$.
%
\end{definition}

\begin{example}
\label{ex:standard interpretation}
We have
$\AllInst{\CTerm{\SET{x}}{\mathsf{f}(x,z)}{y}{x = y \times 2}}
= \SET{ \mathsf{f}(n,w) \mid n \in \mathbb{Z}/2, w \in \xT(\xF,\xV) }$;
thus, e.g., 
$\mathsf{f}(0,\mathsf{f}(x,y)) 
\in \AllInst{\CTerm{\SET{x}}{\mathsf{f}(x,z)}{y}{x = y \times 2}}$.
\end{example}

The following are some immediate properties of the interpretation.

\begin{lemma}
\label{lem:immediate properties of interpretation}
Let $\CTerm{X}{s}{\vec{x}}{\varphi}$ be an existentially constrained term.
Then, 
\begin{enumerate}
    \item $\CTerm{X}{s}{\vec{x}}{\varphi}$ is satisfiable iff
$\AllInst{\CTerm{X}{s}{\vec{x}}{\varphi}} \ne \varnothing$, and
    \item the set $\AllInst{\CTerm{X}{s}{\vec{x}}{\varphi}}$
is closed under substitutions.
\end{enumerate}
\end{lemma}

Subsumption and equivalence of constrained terms are characterized in terms
of their interpretation.

\begin{restatable}{lemma}{LemmaSubsumptionByInstances}
\label{lem:subsumption-by-instances}
Let $\CTerm{X}{s}{\vec{x}}{\varphi}$, $\CTerm{Y}{t}{\vec{y}}{\psi}$ be existentially constrained terms.
Then,
\begin{enumerate}
\item $\CTerm{X}{s}{\vec{x}}{\varphi} \subsetsim \CTerm{Y}{t}{\vec{y}}{\psi}$ 
iff $\AllInst{\CTerm{X}{s}{\vec{x}}{\varphi}} \subseteq \AllInst{\CTerm{Y}{t}{\vec{y}}{\psi}}$, and
\item $\CTerm{X}{s}{\vec{x}}{\varphi} \sim \CTerm{Y}{t}{\vec{y}}{\psi}$ 
iff $\AllInst{\CTerm{X}{s}{\vec{x}}{\varphi}} = \AllInst{\CTerm{Y}{t}{\vec{y}}{\psi}}$.
\end{enumerate}
\end{restatable}

The following properties, which also directly follow from the definition
of subsumption and equivalence, are
immediate corollaries of \Cref{lem:subsumption-by-instances}:
\begin{itemize}
    \item 
    $\subsetsim$ forms a quasi-order,
    \item 
    $\sim$ is an equivalence relation,
    \item 
    all unsatisfiable existentially constrained terms are equivalent, and
    \item
    no satisfiable existentially constrained term is equivalent to any unsatisfiable existentially constrained term.
\end{itemize}

The interpretation of 
rewrite rules needs to synchronize instantiations of both sides of each rule.

\begin{definition}[Interpretation of Constrained Rules]
Let $\rho: \CRu{Z}{\ell}{r}{\pi}$ be a left-linear constrained rewrite rule.
We define the \emph{interpretation} $\AllInst{\rho}$ of $\rho$ as: 
$\AllInst{\rho} = 
\SET{ \langle \ell\sigma,r \sigma \rangle \mid
\Dom(\sigma) = \VDom(\sigma) = Z \cap \Var(\ell,r),
\sigma \vDash_\xM \ECO{\pvec{z}'}{\pi}
}$,
where $\SET{\pvec{z}'} = \Var(\pi) \setminus \Var(\ell,r)$.
For an LCTRS $\xR$,
we define
$\AllInst{\xR} = \bigcup_{\rho \in \xR} \AllInst{\rho}$.
\end{definition}

\begin{example}
\label{exp:intepretation of rewrite rules}
Let $\rho: \CRu{\SET{x}}{\mathsf{f}(x,y)}{y}{0 \ge x}$ be a constrained rewrite rule.
Then 
$\AllInst{\rho} =  \SET{ \mathsf{f}(n,y) \to y \mid 0 \ge n \in \mathbb{Z} }$.
\end{example}

We note 
that, unlike the interpretation of existentially constrained terms, the
non-logical variables in rewrite rules are not instantiated in the
interpretation of rewrite rules.
We define it this way because variable instantiation is built into the rewrite steps.

We continue with characterizations of most general rewrite steps
and partial rewrite steps in terms of their interpretation. We begin with a
characterization of partial rewrite steps.

\begin{restatable}{lemma}{LemmaCharacterizationOfWeakRewriteSteps}
\label{lem:Characterization of Weak Rewrite Steps}
Let $\CTerm{X}{s}{\vec{x}}{\varphi} \leadsto^p_\rho \CTerm{Y}{t}{\vec{y}}{\psi}$.
Then,
\begin{enumerate}
\item there exists $u \in \AllInst{\CTerm{X}{s}{\vec{x}}{\varphi}}$
such that 
$u \to^p_{\AllInst{\rho}} v$ for some $v$,
\item 
if $w \in \AllInst{\CTerm{Y}{t}{\vec{y}}{\psi}}$
then 
there exists $u \in \AllInst{\CTerm{X}{s}{\vec{x}}{\varphi}}$
such that 
$u \to^p_{\AllInst{\rho}} w$, and

\item 
for any $u,v$
such that $u \in \AllInst{\CTerm{X}{s}{\vec{x}}{\varphi}}$
and 
$u \to^p_{\AllInst{\rho}} v$,
we have 
$v \in \AllInst{\CTerm{Y}{t}{\vec{y}}{\psi}}$.
\end{enumerate}
\end{restatable}

\Cref{lem:Characterization of Weak Rewrite Steps} immediately yields
the following characterization of
partial constrained reductions in terms of interpretation.

\begin{restatable}[Characterization by Interpretation for Partial Constrained Reduction]{theorem}{TheoremCharacterizationbyInterpreationForWeakReduction}
\label{thm:Characterization by Interpreation for Weak Reduction}
Let $\CTerm{X}{s}{\vec{x}}{\varphi} \stackrel{*}{\leadsto}_\xR  \CTerm{Y}{t}{\vec{y}}{\psi}$.
Then,  
both of the following statements hold:
\begin{enumerate}
\item
for each $v \in \AllInst{\CTerm{Y}{t}{\vec{y}}{\psi}}$
there exists $u \in \AllInst{\CTerm{X}{s}{\vec{x}}{\varphi}}$
such that $u \stackrel{*}{\to}_{\AllInst{\xR}} v$, and
\item
there exist $u \in \AllInst{\CTerm{X}{s}{\vec{x}}{\varphi}}$ 
and $v \in \AllInst{\CTerm{Y}{t}{\vec{y}}{\psi}}$
such that $u \stackrel{*}{\to}_{\AllInst{\xR}} v$.
\end{enumerate}
\end{restatable}

Note that 
by \Cref{thm:most general rewrite step implies weak rewrite step},
the properties that hold for partial rewrite steps
should hold also for most general rewrite steps.
Thus, 
the properties presented
in \Cref{thm:Characterization by Interpreation for Weak Reduction}
also hold for most general rewrite steps.

In fact, 
a property stronger than (2) holds for most general rewrite steps.
To show this, the following lemma is used.

\begin{restatable}{lemma}{LemmaCharacterizationOfMGRewriteSteps}
\label{lem:Characterization of MG Rewrite Steps}
Let $\CTerm{X}{s}{\vec{x}}{\varphi} \to^p_\rho \CTerm{Y}{t}{\vec{y}}{\psi}$.
For any 
$u \in \AllInst{\CTerm{X}{s}{\vec{x}}{\varphi}}$
there exists $v$ such that
$u \to^p_{\AllInst{\rho}} v$.
\end{restatable}

We now arrive at a characterization of most general constrained reductions
via interpretations.

\begin{restatable}[Characterization by Interpretation for Most General Constrained Reduction]{theorem}{TheoremCharacterizationByInterpreationForMostGeneralReduction}
\label{thm:Characterization by Interpreation for Most General Reduction}
Let $\CTerm{X}{s}{\vec{x}}{\varphi} \stackrel{*}{\to}_\xR  \CTerm{Y}{t}{\vec{y}}{\psi}$.
Then,
\begin{enumerate}
\item
for each $u \in \AllInst{\CTerm{X}{s}{\vec{x}}{\varphi}}$
there exists $v \in \AllInst{\CTerm{Y}{t}{\vec{y}}{\psi}}$
such that $u \stackrel{*}{\to}_{\AllInst{\xR}} v$, and
\item
for each $v \in \AllInst{\CTerm{Y}{t}{\vec{y}}{\psi}}$
there exists $u \in \AllInst{\CTerm{X}{s}{\vec{x}}{\varphi}}$
such that $u \stackrel{*}{\to}_{\AllInst{\xR}} v$.
\end{enumerate}
\end{restatable}

\begin{proof}
\Bfnum{1} follows from \Cref{lem:Characterization of MG Rewrite Steps},
and~\Bfnum{2} follows from 
\Cref{thm:most general rewrite step implies weak rewrite step,thm:Characterization by Interpreation for Weak Reduction}.
\end{proof}

Note that 
the property of \Cref{thm:Characterization by Interpreation for Most General Reduction}~\Bfnum{1}
together with consistency
implies 
the property of \Cref{thm:Characterization by Interpreation for Weak Reduction}~\Bfnum{2}.
On the other hand, 
the property of \Cref{thm:Characterization by Interpreation for Most General Reduction}~\Bfnum{1}
does not hold for partial rewrite steps, as shown in the following example.

\begin{example}
\label{exp:characterization of weak rewrite step I}
Consider \Cref{exp: weak redex and weak rewrite steps}.
We have
\[
\begin{array}{@{}l@{\>}c@{\>}l@{}}
\lefteqn{\CTerm{\SET{y}}{1 + \mathsf{sum}(y)}{w}{1 \le w \land w \le 5 \land y = w - 1}}\\
& \leadsto_{\rho}&
\CTerm{\varnothing}{1 + 0}{w,y}{1 \le w \land w \le 5 \land y = w - 1 \land 0 \ge y}   
\end{array}
\]
for $\rho: \CRu{\SET{x}}{\mathsf{sum}(x)}{0}{0 \ge x}$.
Let $s = 1 + \mathsf{sum}(1)$.
Then 
$s \in 
\AllInst{\CTerm{\SET{y}}{1 + \mathsf{sum}(y)}{w}{1 \le w \land w \le 5 \land y = w - 1}}$.
However, 
we have
$\AllInst{\CTerm{\varnothing}{1 + 0}{w,y}{1 \le w \land w \le 5 \land y = w - 1 \land 0 \ge y}}
=  \SET{ 1 + 0 }$,
and 
$1 + \mathsf{sum}(1)
\to_\rho 1 + 0$ does not hold.
In fact, $s$ is a normal form with respect to the rule $\rho$.
\end{example}

Thus, 
the property of \Cref{thm:Characterization by Interpreation for Most General Reduction}~\Bfnum{1}
discriminates most general rewrite steps
from partial rewrite steps.
Then, a natural question arise:
Is there any property that holds for partial rewrite steps,
but not for most general rewrite steps?

\begin{example}
\label{exp:characterization of weak rewrite step II}
Consider \Cref{exp:characterization of weak rewrite step I}.
We have
$\AllInst{\CTerm{\SET{y}}{1 + \mathsf{sum}(y)}{w}{1 \le w \land w \le 5 \land y = w - 1}}
= \SET{ 1 + \mathsf{sum}(n) \mid 0 \le n \le 4 }$.
The reason for the partial rewrite step
by the rule $\rho: \CRu{\SET{x}}{\mathsf{sum}(x)}{0}{0 \ge x}$
is the existence of $1 + \mathsf{sum}(0)$ in this set.
On the other hand,
$s' = \CTerm{\SET{y}}{1 + \mathsf{sum}(w)}{}{1 \le w \land w \le 5}$
is a normal form with respect to the rule $\rho$.
In this case, 
we have $\AllInst{s'} = \SET{ 1 + \mathsf{sum}(n) \mid 1 \le n \le 5 }$, whose members are all 
normal forms with respect to the rule $\rho$.
This motivates us to characterize normal forms of partial write steps
by the absence of reducible terms in interpretations.
However, the interpretation so far is not useful for this.
Consider the rule $\rho': \CRu{\SET{x}}{\mathsf{f}(x,y)}{y}{0 = x}$
and the existentially constrained term 
$t = \CTerm{\SET{w}}{\mathsf{f}(w,z)}{}{1 \le w \land w \le 5}$.
Then, although $t$ is normal with respect to the rule $\rho'$,
$\AllInst{t}$ includes
the term $\mathsf{f}(1,\mathsf{f}(0,z))$,
which is reducible by the rule $\rho'$.
For a different illustration, 
consider the following rewrite rule 
$\rho'': \Pi \varnothing. \mathsf{f}(\mathsf{a}) \to \mathsf{b}[\mathsf{true}]$
and the constrained term $\CTerm{\varnothing}{\mathsf{f}(x)}{}{\mathsf{true}}$.
On the one hand
$\mathsf{f}(\mathsf{a}) \in \AllInst{\CTerm{\varnothing}{\mathsf{f}(x)}{}{\mathsf{true}}}$
is reducible, while on the other hand
$\mathsf{f}(\mathsf{b}) \in \AllInst{\CTerm{\varnothing}{\mathsf{f}(x)}{}{\mathsf{true}}}$
is a normal form.
Thus the standard interpretation can hardly be used to distinguish
normal form instances from reducible instances.
\end{example}

In the following sections, we give a
characterization of partial rewrite steps
which does not hold for most general rewrite steps,
by introducing another kind of interpretations.

\section{%
Value Interpretation of Existentially Constrained Terms}
\label{sec:value-interpretation}

The standard interpretation of existentially constrained terms
in \Cref{sec:Interpretations of Weak and Most General
Reductions} is inadequate for
characterizing partial rewrite steps by their normal forms.
In this section, we introduce an alternative interpretation of 
constrained terms. The idea is that logical variables are
instantiated by values, as before, but for non-logical variables we only allow
renaming instead of proper instantiation.

\begin{definition}[Value Interpretation of Constrained Terms]
A term $t$ is said to be a \emph{value instance} of 
an existentially constrained term
$\CTerm{X}{s}{\vec{x}}{\varphi}$ if 
$t=s\sigma$ for some substitution $\sigma$ such that 
\begin{enumerate}
\item 
$\sigma(X) \subseteq \Val$,
\item 
$\sigma(\Var(s) \setminus X) \subseteq \xV$,
\item 
$x \neq y$ implies $\sigma(x) \neq \sigma(y)$ for any variable 
$x,y \in \Var(s) \setminus X$, and
\item 
$\vDash_{\xM,\sigma} \ECO{\vec{x}}{\varphi}$.
\end{enumerate}
We denote the set of all value instances of $\CTerm{X}{s}{\vec{x}}{\varphi}$ by
$\AllInst{\CTerm{X}{s}{\vec{x}}{\varphi}}_\mathsf{v}$,
which is called the \emph{value interpretation} of 
the constrained term $\CTerm{X}{s}{\vec{x}}{\varphi}$.
\end{definition}
Note here that, as $\FVar(\ECO{\vec{x}}{\varphi}) \subseteq X$,
it follows from the items~\Bfnum{1} and~\Bfnum{3}
that $\sigma \vDash_{\xM} \ECO{\vec{x}}{\varphi}$.
The reason to allow renaming of non-logical variables 
for value interpretations is to enable it to
characterize equivalence ($\sim$) between (existentially) constrained terms
(\cref{thm:characterization-of-equivalence-between-constrained-terms-by-value-interpretations});
see the note after \Cref{ex:counter example for reflection of subsumption in value intepretation}.

\begin{example}
\label{ex:value interpretation}
We have
$\AllInst{\CTerm{\SET{x}}{\mathsf{f}(x,z)}{y}{x = y \times 2}}_\mathsf{v}
= \SET{ \mathsf{f}(n,w) \mid n \in \mathbb{Z}/2, w \in \xV }$.
Thus, $\mathsf{f}(0,x),\mathsf{f}(2,y),\ldots \in 
\AllInst{\CTerm{\SET{x}}{\mathsf{f}(x,z)}{y}{x = y \times 2}}_\mathsf{v}$.
Unlike the standard interpretation,
$\mathsf{f}(0,\mathsf{f}(x,y))
\notin \AllInst{\CTerm{\SET{x}}{\mathsf{f}(x,z)}{y}{x = y \times 2}}_\mathsf{v}$
(cf.\ \Cref{ex:standard interpretation}).
\end{example}

Some immediate properties of value instantiations follow.

\begin{lemma}
\label{lem:simple properties of value instantiation}
\Bfnum{1.}
An existentially constrained term $\CTerm{X}{s}{\vec{x}}{\varphi}$ is satisfiable 
iff $\AllInst{\CTerm{X}{s}{\vec{x}}{\varphi}}_\mathsf{v} \ne \varnothing$.
\Bfnum{2.}
The set $\AllInst{\CTerm{X}{s}{\vec{x}}{\varphi}}_\mathsf{v}$
is closed under variable renaming.
\Bfnum{3.}
$\AllInst{\CTerm{X}{s}{\vec{x}}{\varphi}}_\mathsf{v}
\subseteq \AllInst{\CTerm{X}{s}{\vec{x}}{\varphi}}$.
\end{lemma}

We focus now on less obvious properties of value instantiations.
We first clarify notable relations between value instantiations and instantiations.
The standard interpretation is obtained from 
the value interpretation in the following way.

\begin{restatable}{lemma}{LemmaObtainAllInstFromValueInst}
\label{lem:obtain all inst from value inst}
For any existentially constrained term
$\CTerm{X}{s}{\vec{x}}{\varphi}$, we have the interpretation
$\AllInst{\CTerm{X}{s}{\vec{x}}{\varphi}}
= \SET{ t\theta \mid 
t \in \AllInst{\CTerm{X}{s}{\vec{x}}{\varphi}}_\mathsf{v},
\mbox{$\theta$ is a substitution} }$.
\end{restatable}

From now on we write $s \leqslant t$ if $s\theta = t$.
Conversely, the value interpretation can be obtained from 
the standard interpretation as follows.

\begin{restatable}{lemma}{LemmaObtainValueInstFromAllInst}
\label{lem: obtain value inst from all inst}
Let $\CTerm{X}{s}{\vec{x}}{\varphi}$ be an existentially constrained term.
Then, the set of terms $\AllInst{\CTerm{X}{s}{\vec{x}}{\varphi}}_\mathsf{v}$
is minimal with respect to $\leqslant$ in 
$\AllInst{\CTerm{X}{s}{\vec{x}}{\varphi}}$.
\end{restatable}


From \Cref{lem:obtain all inst from value inst,lem: obtain value inst from all inst}, it follows
that equivalence of constrained terms coincides
with the equality of their value instantiations,
similarly with that of their standard instantiations.

\begin{restatable}{theorem}{TheoremCharacterizationOfEquivalenceBetweenConstrainedTermsByValueInterpretations}
\label{thm:characterization-of-equivalence-between-constrained-terms-by-value-interpretations}
Let $\CTerm{X}{s}{\vec{x}}{\varphi}$, $\CTerm{Y}{t}{\vec{y}}{\psi}$ be existentially constrained terms.
Then,
$\CTerm{X}{s}{\vec{x}}{\varphi} \sim \CTerm{Y}{t}{\vec{y}}{\psi}$ 
iff $\AllInst{\CTerm{X}{s}{\vec{x}}{\varphi}}_\mathsf{v}
 = \AllInst{\CTerm{Y}{t}{\vec{y}}{\psi}}_\mathsf{v}$.
\end{restatable}

\begin{proof}
We show that
$\AllInst{\CTerm{X}{s}{\vec{x}}{\varphi}}
 = \AllInst{\CTerm{Y}{t}{\vec{y}}{\psi}}$
iff $\AllInst{\CTerm{X}{s}{\vec{x}}{\varphi}}_\mathsf{v}
 = \AllInst{\CTerm{Y}{t}{\vec{y}}{\psi}}_\mathsf{v}$.
The \textit{if}-part follows from \Cref{lem:obtain all inst from value inst},
while the \textit{only-if} part follows from \Cref{lem: obtain value inst from all inst}.
The full claim then follows from \Cref{lem:subsumption-by-instances}~\Bfnum{2}.
\end{proof}

Contrast to the equivalence $\sim$ of existentially constraint
terms, the value interpretation does not reflect subsumption relation on it,
i.e., 
the following is not correct:
$\CTerm{X}{s}{\vec{x}}{\varphi} \subsetsim \CTerm{Y}{t}{\vec{y}}{\psi}$ 
iff $\AllInst{\CTerm{X}{s}{\vec{x}}{\varphi}} \subseteq \AllInst{\CTerm{Y}{t}{\vec{y}}{\psi}}$.
This is witnessed in the following example.

\begin{example}
\label{ex:counter example for reflection of subsumption in value intepretation}
Consider the two existential constrained terms
$\CTerm{\SET{ x }}{\mathsf{f}(x)}{}{\mathsf{true}}$
and
$\CTerm{\varnothing}{\mathsf{f}(x)}{}{\mathsf{true}}$.
Then on the one hand we have
$\CTerm{\SET{ x }}{\mathsf{f}(x)}{}{\mathsf{true}}
\subsetsim 
\CTerm{\varnothing}{\mathsf{f}(x)}{}{\mathsf{true}}$.
On the other hand,
$\AllInst{\CTerm{\SET{ x }}{\mathsf{f}(x)}{}{\mathsf{true}}}_\mathsf{v}
\subseteq 
\AllInst{\CTerm{\varnothing}{\mathsf{f}(x)}{}{\mathsf{true}}}_\mathsf{v}$
does not hold.
It gets more obvious by $\AllInst{\CTerm{\SET{ x }}{\mathsf{f}(x)}{}{\mathsf{true}}}_\mathsf{v}
= \SET{ \mathsf{f}(v) \mid v \in  \Val }
\neq \SET{ \mathsf{f}(y) \mid y \in  \xV }
= \AllInst{\CTerm{\varnothing}{\mathsf{f}(x)}{}{\mathsf{true}}}_\mathsf{v}$.
\end{example}

Note that if one defines the value interpretation
without allowing to rename non-logical variables, 
then \Cref{thm:characterization-of-equivalence-between-constrained-terms-by-value-interpretations}
does not hold.
Consider 
$\CTerm{\varnothing}{\mathsf{f}(x)}{}{\mathsf{true}}
\sim  \CTerm{\varnothing}{\mathsf{f}(y)}{}{\mathsf{true}}$
for which we obtain
$\AllInst{\CTerm{\varnothing}{\mathsf{f}(x)}{}{\mathsf{true}}}_\mathsf{v}
= \SET{ \mathsf{f}(x) } \neq \SET{ \mathsf{f}(y) } =
\AllInst{\CTerm{\varnothing}{\mathsf{f}(y)}{}{\mathsf{true}}}_\mathsf{v}$.

%
%


\section{Characterization of Partial and Most General Constrained Reductions 
by Value Instantiation}
\label{sec:characterization}

In this section, we provide a characterization of partial constrained reductions and most
general constrained reductions with respect to their normal forms, expressed in terms of
value interpretation. In contrast to the characterizations in
\Cref{sec:Interpretations of Weak and Most General Reductions}, this
characterization distinguishes partial rewrite steps from most general rewrite steps.
Firstly, we extend \Cref{lem:Characterization of Weak Rewrite Steps}
to value interpretation.

\begin{restatable}{lemma}{LemmaCharacterizationOfWeakRewriteStepsByValueInterpretation}
\label{lem:Characterization of Weak Rewrite Steps by Value Interpretation}
Suppose $\CTerm{X}{s}{\vec{x}}{\varphi} \leadsto^p_\rho \CTerm{Y}{t}{\vec{y}}{\psi}$.
Then all of the following holds:
\begin{enumerate}
\item 
There exists $u \in \AllInst{\CTerm{X}{s}{\vec{x}}{\varphi}}_\mathsf{v}$
such that $u \to^p_{\AllInst{\rho}} v$ for some $v$.

\item 
If $w \in \AllInst{\CTerm{Y}{t}{\vec{y}}{\psi}}_\mathsf{v}$
then there exists $u \in \AllInst{\CTerm{X}{s}{\vec{x}}{\varphi}}_\mathsf{v}$
such that $u \to^p_{\AllInst{\rho}} w$.

\item 
If $u \in \AllInst{\CTerm{X}{s}{\vec{x}}{\varphi}}_\mathsf{v}$
and $u \to^p_{\AllInst{\rho}} v$,
then $v \in \AllInst{\CTerm{Y}{t}{\vec{y}}{\psi}}_\mathsf{v}$.
\end{enumerate}
\end{restatable}

Now we proceed with a characterization in the reverse direction,
that is, when a partial rewrite step from a constrained term
is guaranteed by the rewrite step of its instantiation.

\begin{restatable}{lemma}{LemmaNormalFormsOfWeakReduction}
\label{lem:Normal Forms of Weak Reduction}
Let $\CTerm{X}{s}{\vec{x}}{\varphi}$ be an existentially constrained term.
Suppose there exists 
$u \in \AllInst{\CTerm{X}{s}{\vec{x}}{\varphi}}_\mathsf{v}$
such that 
$u \to^p_{\AllInst{\rho}} v$
for some position $p$ and a term $v$.
Then, 
there exists an existentially constrained term $\CTerm{Y}{t}{\vec{y}}{\psi}$ such that
$\CTerm{X}{s}{\vec{x}}{\varphi} \leadsto^p_\rho \CTerm{Y}{t}{\vec{y}}{\psi}$
and 
$v \in \AllInst{\CTerm{Y}{t}{\vec{y}}{\psi}}_\mathsf{v}$.
\end{restatable}

%
%

\begin{restatable}[Normal Forms of Partial Constrained Reduction]{theorem}{TheoremNormalFormsOfWeakReduction}
\label{thm:Normal Forms of Weak Reduction}
Suppose $\xR$ is a left-linear and left-value-free LCTRS.
Let $\CTerm{X}{s}{\vec{x}}{\varphi}$ be an existentially constrained term.
Then, 
$\CTerm{X}{s}{\vec{x}}{\varphi}$
is a normal form with respect to ${\leadsto}_\xR$
iff
any $u \in \AllInst{\CTerm{X}{s}{\vec{x}}{\varphi}}_\mathsf{v}$
is a normal form with respect to  $\to_{\AllInst{\xR}}$.
\end{restatable}

\begin{proof}
($\Longleftarrow$)
We show the contraposition.
Suppose $\CTerm{X}{s}{\vec{x}}{\varphi} \leadsto^p_\rho \CTerm{Y}{t}{\vec{y}}{\psi}$.
From 
\Cref{lem:Characterization of Weak Rewrite Steps by Value Interpretation}~\Bfnum{1},
there exists 
$u \in \AllInst{\CTerm{X}{s}{\vec{x}}{\varphi}}_{\mathsf{v}}$ and 
$v$ such that 
$u \to^p_{\AllInst{\rho}} v$.
($\Longrightarrow$)
We proceed with a proof by contradiction.
Assume there exists $u \in \AllInst{\CTerm{X}{s}{\vec{x}}{\varphi}}_\mathsf{v}$
such that $u \to_{\AllInst{\xR}}^p v$ for some $v$.
Then, by the definition of value instantiation,
$u = s\mu$ for some $X$-valued substitution $\mu$
such that $\Dom(\mu) = X$ and $\mu \vDash_\xM \ECO{\vec{x}}{\varphi}$.
By $u \to_{\AllInst{\xR}}^p v$,
we have $\langle u|_p, v|_p \rangle \in \AllInst{\rho}$ for some $\rho \in \xR$.
However, it follows from \Cref{lem:Normal Forms of Weak Reduction}
that $\CTerm{X}{s}{\vec{x}}{\varphi} \leadsto^p_\rho \CTerm{Y}{t}{\vec{y}}{\psi}$.
\end{proof}

Note that the theorem above can be reformulated as follows.
The following statements are equivalent:
(i) $\CTerm{X}{s}{\vec{x}}{\varphi} \leadsto_\xR \CTerm{Y}{t}{\vec{y}}{\psi}$
for some $\CTerm{Y}{t}{\vec{y}}{\psi}$, and
(ii)
there exists some $u \in \AllInst{\CTerm{X}{s}{\vec{x}}{\varphi}}_\mathsf{v}$
such that $u \to_{\AllInst{\xR}} v$ for some $v$.

\begin{example}
\label{exp:characterization of weak rewrite step III}
The ($\Longrightarrow$)-direction of \Cref{thm:Normal Forms of Weak Reduction}
does not hold if we replace 
the value interpretation $\AllInst{\CTerm{X}{s}{\vec{x}}{\varphi}}_\mathsf{v}$
with the standard interpretation  $\AllInst{\CTerm{X}{s}{\vec{x}}{\varphi}}$.
To see this, let us revisit \Cref{exp:characterization of weak rewrite step II}.
Consider 
the rule $\rho': \CRu{\SET{x}}{\mathsf{f}(x,y)}{y}{0 = x}$
and the existentially constrained term 
$t = \CTerm{\SET{w}}{\mathsf{f}(w,z)}{}{1 \le w \land w \le 5}$.
Then $\mathsf{f}(1,\mathsf{f}(0,z)) \in \AllInst{t}$
and $\mathsf{f}(1,\mathsf{f}(0,z)) \to^2_{\AllInst{\rho'}}  \mathsf{f}(1,z)$
by the rule $\mathsf{f}(\mathsf{0},y) \to y \in \AllInst{\rho'}$.
However, there does not exist an $s$ such that $t \leadsto^2_{\rho'} s$.
Consider 
the rule $\rho'': \Pi \varnothing. \mathsf{f}(\mathsf{a}) \to \mathsf{b}[\mathsf{true}]$
and the constrained term $t' = \CTerm{\varnothing}{\mathsf{f}(x)}{}{\mathsf{true}}$.
Then 
$\mathsf{f}(\mathsf{a}) \in \AllInst{t'}$
and 
$\mathsf{f}(\mathsf{a}) \to^\epsilon_{\AllInst{\rho''}} \mathsf{b}$
by using the rule $\mathsf{f}(\mathsf{a}) \to \mathsf{b} \in \AllInst{\rho''}$.
However, there does not exist an $s'$ such that 
$t' \leadsto^\epsilon_{\rho''} s'$.
\end{example}

Under the view of \Cref{thm:Characterization by Interpreation for Weak Reduction,thm:Characterization by Interpreation for Most General Reduction},
one might conjecture that
if one replaces ``any $u$'' in the statement of \Cref{thm:Normal Forms of Weak Reduction}
by ``there exists $u$'',
then the normal form of constrained terms with respect to most general constrained reduction
can be obtained.
However, this is not the case.
Consider $\mathcal{R}$ in \Cref{exp:strong rewrite step}
and the constrained term $\mathsf{sum}(x)~[0 \le x \land x \le 1]$.
Then $\AllInst{\mathsf{sum}(x)~[0 \le x \land x \le 1]}_\mathsf{v}
= \SET{ \mathsf{sum}(0), \mathsf{sum}(1) }$
and we have $\mathsf{sum}(0) \to_{\AllInst{\xR}} 0$
and $\mathsf{sum}(1) \to_{\AllInst{\xR}} 1 + \mathsf{sum}(1 + -1)$.
However, 
$\mathsf{sum}(x)~[0 \le x \land x \le 1]$
is a normal form with respect to $\to_\xR$.

We are now going to show a variation of \Cref{thm:Normal Forms of Weak Reduction}
for the most general rewrite step, by 
modifying the form of the necessary and sufficient condition of the theorem.
To aim for this, we first present a variant of 
\Cref{lem:Normal Forms of Weak Reduction} for most general rewrite steps, 
which can be shown with a help of
additional condition that the rewrite rules is left-value-free.

\begin{restatable}{lemma}{LemmaInstantiationsAreRedexIsEnoughForRewriteSteps}
\label{lem:instantiations are redex is enough for rewrite steps}
Let $\CTerm{X}{s}{\vec{x}}{\varphi}$ be a satisfiable existentially constrained term
and $\rho$ a left-value-free and left-linear rewrite rule.
Let $p \in \Pos(s)$.
Suppose that for any $u \in \AllInst{\CTerm{X}{s}{\vec{x}}{\varphi}}_\mathsf{v}$
there exists $v$ such that
$u \to^p_{\AllInst{\rho}} v$.
Then 
$\CTerm{X}{s}{\vec{x}}{\varphi}
\to^p_\rho
\CTerm{Y}{t}{\vec{y}}{\psi}$
for some $\CTerm{Y}{t}{\vec{y}}{\psi}$.
\end{restatable}

As a side product, now it follows that 
\Cref{lem:Characterization of MG Rewrite Steps}
can be strengthened with the converse direction
for the case that the rewrite rule
is left-value-free.

\begin{restatable}{corollary}{CorollaryCharacterizationOfMGRewriteSteps}
Let $\CTerm{X}{s}{\vec{x}}{\varphi}$ be an existentially constrained term
and $\rho$ a left-linear and left-value-free rewrite rule.
Let $p \in \Pos(s)$.
Then the following statements are equivalent:
\begin{enumerate}
\item
for any $u \in \AllInst{\CTerm{X}{s}{\vec{x}}{\varphi}}$
there exists $v$ such that
$u \to^p_{\AllInst{\rho}} v$,
and
\item 
$\CTerm{X}{s}{\vec{x}}{\varphi}
\to^p_\rho
\CTerm{Y}{t}{\vec{y}}{\psi}$
for some  $\CTerm{Y}{t}{\vec{y}}{\psi}$.
\end{enumerate}
\end{restatable}

\begin{proof}
(\Bfnum{2} $\Longrightarrow$ \Bfnum{1}) follows by \Cref{lem:Characterization of MG Rewrite Steps}.
(\Bfnum{1} $\Longrightarrow$ \Bfnum{2})
follows
from 
\Cref{lem:instantiations are redex is enough for rewrite steps}
and 
\Cref{lem:simple properties of value instantiation}~\Bfnum{3}.
\end{proof}

Now we arrive at 
a characterization of normal forms with respect
to most general constrained reduction
in terms of value interpretation.

\begin{restatable}[Normal Forms of Most General Constrained Reduction]{theorem}{TheoremNormalFormsOfMostGeneralReduction}
\label{thm:Normal Forms of Most General Reduction}
Suppose $\xR$ is a left-linear and left-value-free LCTRS.
Let $\CTerm{X}{s}{\vec{x}}{\varphi}$ be a satisfiable existentially constrained term.
Then, 
$\CTerm{X}{s}{\vec{x}}{\varphi}$ is a normal form with respect to 
$\to_{\xR}$
iff
for any $\rho \in \xR$ and position $p \in \Pos(s)$,
there exists $u \in \AllInst{\CTerm{X}{s}{\vec{x}}{\varphi}}_\mathsf{v}$
such that $u \to^p_{\AllInst{\rho}} v$ for no $v$.
\end{restatable}

\begin{proof}
($\Longleftarrow$)
We show the contraposition.
Suppose $\CTerm{X}{s}{\vec{x}}{\varphi} \to_{\xR} \CTerm{Y}{t}{\vec{y}}{\psi}$.
Then,
$\CTerm{X}{s}{\vec{x}}{\varphi} \to^p_\rho \CTerm{Y}{t}{\vec{y}}{\psi}$
for some $p \in \Pos_\xF(s)$ and $\rho \in \xR$.
Let $u \in \AllInst{\CTerm{X}{s}{\vec{x}}{\varphi}}_\mathsf{v}$.
Then, as $\AllInst{\CTerm{X}{s}{\vec{x}}{\varphi}}_\mathsf{v}
\subseteq \AllInst{\CTerm{X}{s}{\vec{x}}{\varphi}}$
from \Cref{lem:simple properties of value instantiation}~\Bfnum{3},
we know $u \in \AllInst{\CTerm{X}{s}{\vec{x}}{\varphi}}$.
By \Cref{lem:Characterization of MG Rewrite Steps},
it follows from $\CTerm{X}{s}{\vec{x}}{\varphi} \to^p_\rho \CTerm{Y}{t}{\vec{y}}{\psi}$
that there exists $v$ such that $u \to^p_{\AllInst{\rho}} v$.
($\Longrightarrow$)
We prove the contraposition.
Suppose there exist $\rho \in \xR$ and $p \in \Pos(s)$
such that
for any $u \in \AllInst{\CTerm{X}{s}{\vec{x}}{\varphi}}_\mathsf{v}$
$u \to^p_{\AllInst{\rho}} v$ for some $v$.
Then by \Cref{lem:instantiations are redex is enough for rewrite steps},
$\CTerm{X}{s}{\vec{x}}{\varphi} \to^p_\rho \CTerm{Y}{t}{\vec{y}}{\psi}$
for some $\CTerm{Y}{t}{\vec{y}}{\psi}$.
\end{proof}

Note that the theorem above can be reformulated as follows.
The following statements are equivalent:
(i) $\CTerm{X}{s}{\vec{x}}{\varphi} \to_{\xR} \CTerm{Y}{t}{\vec{y}}{\psi}$
for some $\CTerm{Y}{t}{\vec{y}}{\psi}$, and
(ii) there exists $\rho \in \xR$ and position $p \in \Pos(s)$
such that for any $u \in \AllInst{\CTerm{X}{s}{\vec{x}}{\varphi}}_\mathsf{v}$
there exists $v$ such that $u \to^p_{\AllInst{\rho}} v$.
In contrast to the characterization of
partial constrained reduction (\Cref{thm:Normal Forms of Weak Reduction}),
our characterization of most general constrained reduction
refers to 
reducibility of all $u \in \AllInst{\CTerm{X}{s}{\vec{x}}{\varphi}}_\mathsf{v}$
with respect to a fixed position and a fixed rewrite rule.


\begin{example}
\label{exp:strong rewrite step II}
Let us recall the LCTRS $\mathcal{R}$ in \Cref{exp:strong rewrite step}:
\[
\mathcal{R} = 
\left\{
\begin{array}{l@{\>}c@{\>}ll}
\rho_1: \Pi \SET{x}.~\mathsf{sum}(x) & \to&  0 &  [0 \ge x ]\\
\rho_2: \Pi \SET{x}.~\mathsf{sum}(x) &\to&  x + \mathsf{sum}(x + -1) &   [ 0 < x ]
\end{array}
\right
\}
\]
Consider the rewrite steps starting from the constrained term
$\CTerm{\SET{x}}{\mathsf{sum}(x)}{}{0 \le x \land x \le 4}$.
Note that we have
$\AllInst{\CTerm{\SET{x}}{\mathsf{sum}(x)}{}{0 \le x \land x \le 4}}_\mathsf{v}
= \SET{ \mathsf{sum}(0),\ldots, \mathsf{sum}(4) }$,
$\AllInst{\rho_1} = \SET{ \mathsf{sum}(n)  \to  0 \mid 0 \ge n \in \mathbb{Z}}$,
and 
$\AllInst{\rho_2} = \SET{ \mathsf{sum}(n)  \to n +  \mathsf{sum}(n + -1) 
\mid 0 < n \in \mathbb{Z} }$.
We further have 
$\CTerm{\SET{x}}{\mathsf{sum}(x)}{}{0 \le x \land x \le 4}
\leadsto^\epsilon_{\rho_1}
\CTerm{\SET{x}}{0}{}{0 \le x \land x \le 4 \land 0 \ge x}$,
and
$\CTerm{\SET{x}}{\mathsf{sum}(x)}{}{0 \le x \land x \le 4}
\leadsto^\epsilon_{\rho_2}
\CTerm{\SET{x}}{x + \mathsf{sum}(x + -1)}{}{0 \le x \land x \le 4 \land 0 < x}$.
By \Cref{thm:Normal Forms of Weak Reduction},
the former is characterized by 
$\mathsf{sum}(0) \to^\epsilon_{\AllInst{\rho}} 0$
and the latter by 
$\mathsf{sum}(4) \to^\epsilon_{\AllInst{\rho}} 4 + \mathsf{sum}(4 + -1)$.
Also none of the TRSs $\AllInst{\rho_1},\AllInst{\rho_2}$ can 
reduce all of the terms $\mathsf{sum}(0),\ldots, \mathsf{sum}(4)$.
By \Cref{thm:Normal Forms of Most General Reduction},
we obtain that $\CTerm{\SET{x}}{\mathsf{sum}(x)}{}{0 \le x \land x \le 4}$
is a normal form with respect to $\to_\mathcal{R}$.
\end{example}

In~\cite[Definition~1]{SM25},
a notion of ``normal form'' for constrained terms
is defined by the condition that all their standard instantiations are
normal with respect to rewriting (non-constrained) terms.
Let us call a constrained term $\CTerm{X}{s}{\vec{x}}{\varphi}$ 
\emph{instantiation-normal} (with respect to an LCTRS $\xR$)
if $u$ is normal with respect to the rewrite step $\to_\xR$ for any (non-constrained) terms
$u \in \AllInst{\CTerm{X}{s}{\vec{x}}{\varphi}}$.
From this we construct the following claim.

\begin{restatable}{proposition}{PropositionIntantiationNormalImpliesNormalwrtPartialRewriting}
\label{prop: instantiation-normal implies normal w.r.t. partial rewriting}
If $\CTerm{X}{s}{\vec{x}}{\varphi}$ is instantiation-normal,
then $\CTerm{X}{s}{\vec{x}}{\varphi}$ is a normal form with respect to $\leadsto_{\xR}$
(and hence, it is also a normal form with respect to $\to_{\xR}$).
\end{restatable}

The converse of 
\Cref{prop: instantiation-normal implies normal w.r.t. partial rewriting}
does not hold.
Consider $\rho: \mathsf{f}(\mathsf{a}) \to \mathsf{b}~[\mathsf{true}]$.
Then $\CTerm{\varnothing}{\mathsf{f}(x)}{}{\mathsf{true}}$
is a normal form with respect to $\leadsto_{\xR}$.
However, 
$\mathsf{f}(\mathsf{a}) \in
\AllInst{\CTerm{\varnothing}{\mathsf{f}(x)}{}{\mathsf{true}}}$
and for this we have
$\mathsf{f}(\mathsf{a}) \to_{\xR} \mathsf{b}$.
Hence $\CTerm{\varnothing}{\mathsf{f}(x)}{}{\mathsf{true}}$ is not
instantiation-normal.

\section{Related Work}
\label{sec:related-work}

Partial constrained reduction implicitly appears in various forms in 
the applications and the analyses of LCTRSs, and also in
a related formalism of rewriting with built-in data structures. 

\subparagraph*{Critical Pair Analysis}

Critical pair analysis is a standard technique for checking confluence of
LCTRSs. Many approaches rely on closure conditions for constrained critical
pairs. Splitting a constrained critical pair~\cite{SM25} partitions its set of instances
into subsets, each of which may admit a rewrite step. The goal is that, after
splitting, a concrete rewrite step becomes applicable and satisfies the closing
condition, which corresponds to a partial rewrite step on the original constrained
critical pair. Partial rewrite steps can also be used to determine whether such a
split is possible, as exploited to some extent in tools like \crest{}~\cite{SM25}.
Partial rewrite steps can also be used to detect non-confluence in LCTRSs, since it
suffices that an instance of a critical pair diverges. In this case, explicit
splitting is not required: a partial step could directly expose the subset of
instances responsible for divergence. Tools like \crest{} exploit this principle
to efficiently find counterexamples without actual partitioning.

\subparagraph*{Constrained Narrowing}

Narrowing is a symbolic computation that extends rewriting;
unlike rewriting, it may involve a minimal instantiation of the target term by
replacing matching with unification.
The notion of narrowing has been incorporated for constrained terms.
``Constrained Narrowing''~\cite{KN24jip} narrows both, 
the term part and the constraint part,
by instantiation for the former and by addition of constraints to the latter. 
More precisely, to apply a
constrained rewrite rule to a constrained term, variables of the
constrained term are minimally instantiated, 
and the constraint of the rule is added to the 
constraint of the constrained term. Then the rewrite rule is applied.
The constraint on the term is not required to imply that of the rewrite
rule, and its satisfiability is determined through a satisfiability check of the
result. Thus, constrained narrowing can be considered as ``Narrowing'' +
``Partial Rewriting''.


\subparagraph*{Constrained Rewriting Induction}


An equation is said to be an inductive theorem of an equational theory 
when it is valid in the initial algebra of the theory; 
in other words, when all ground instantiations of the equation
are equational consequences of the theory (e.g.,~\cite{HuetOppen}).
Rewriting induction, that has been initiated in~\cite{ReddyRI}, is an 
approach to verify inductive theorems based on rewriting techniques.
The notion of inductive theorems and the rewriting induction approach
have been incorporated to the LCTRS formalism in~\cite{FKN17tocl,KN15},
where the latter has been called ``Constrained Rewriting Induction.''
A key inference rule in rewriting induction is the rule Expand, which
consists of a collection of narrowing steps applied at an appropriate position
in an equation using all rules applicable at that position. 
In constrained rewriting induction, instead of using only pure narrowing
steps, a combination of narrowing steps and partial constrained rewrite steps is
employed.
This approach collects the reducts of all possible rewrite steps for
minimal instantiations, while potentially adding constraints to the constrained
equation at an appropriate position.
In addition, although not covered in~\cite{FKN17tocl},
partial constrained reduction may be useful to 
strengthen the Disprove rule for refuting a conjecture in constrained
rewriting induction, similarly to the part on critical pair analysis above.

\subparagraph*{Rewriting Modulo SMT}
``Rewriting Modulo SMT''~\cite{RMM17} is a sophisticated and intricate
rewriting formalism (in a broad sense) that is based on a framework and
terminology different from that of LCTRSs. Nevertheless, it shares key features
with LCTRSs, such as the ability to handle built-in data structures and the use
of constrained rewrite rules and rewriting on constrained terms.
Although precise comparisons between the two formalisms are beyond the scope
of this paper, roughly speaking, this formalism is designed to model state
transition systems for more application-oriented specifications. In contrast to
LCTRSs, it includes extended features such as order-sorted signatures and
rewriting modulo equations, as well as restrictions like rewriting only at the
root position, focusing on initial models, and a hierarchical signature for term
sorts over theory sorts.
We would like to highlight two noteworthy correspondences
between results in~\cite{RMM17} and ours.
First, the basic rewrite relation (referred to as the symbolic rewrite relation in~\cite{RMM17})
corresponds more closely to partial constrained rewriting
rather than to most general constrained rewriting.
Second, papers establish soundness and completeness
results with respect to the standard interpretations.
Readers may compare the
soundness results (\Cref{lem:Characterization of Weak Rewrite Steps,lem:Characterization of Weak Rewrite Steps by Value Interpretation}
versus \cite[Theorem~1]{RMM17})
and 
completeness results (\Cref{lem:Normal Forms of Weak Reduction,exp:characterization of weak rewrite step III}
versus \cite[Theorem~2]{RMM17}).
Overall, our results help to clarify the differences 
between the rewriting modulo SMT and LCTRS formalisms and highlight
the respective advantages of each.


\section{Conclusion}
\label{sec:conclusion}

In this paper, we have introduced a novel notion of partial constrained
rewriting for existentially constrained terms. The key difference between the
definition of most general constrained rewrite steps and that of partial
constrained rewrite steps lies in the definition of redex---validity is 
needed for most general rewriting while satisfiability
suffices for partial rewriting. We have shown that most general rewrite
steps are partial, while partial rewrite steps can be simulated by the
composition of subsumption, most general rewrite steps and equivalence. We have
further developed rigorous characterizations of partial and most general
constrained reductions through value and standard interpretations. We have
shown that while the standard interpretation captures some behavior of partial
and most general rewrite steps, it is insufficient to distinguish partial
rewrite steps from most general rewrite steps. 
By introducing the notion of value interpretation, we fully describe normal
forms of both types of constrained reductions, including successfully
highlighting subtle differences between partial and most general constrained
reductions. 
We have also revisited immature ideas of partial rewriting that appeared in the literature;
we anticipate that our
results not only clarify the theoretical structure of logically constrained
rewriting but also lay the groundwork for further exploration of 
(non-)reachability, (non-)joinability and (non-)convertibility
verification on logically constrained terms,
that appear in many applications and analyses of LCTRSs.

We focus in this paper on constrained rewriting with left-linear constrained rewrite rules.
Despite some complications, we anticipate that one can extend
partial constrained rewriting to non-left-linear rules.
We leave this extension as future work. Another question is,
whether partial constrained rewriting and equivalence commute.
We expect this to be the case, but solving this problem
also remains as future work.

\bibliography{biblio}

\begin{thebibliography}{10}

\bibitem{ANS24}
Takahito Aoto, Naoki Nishida, and Jonas Sch{\"o}pf.
\newblock Equational theories and validity for logically constrained term
  rewriting.
\newblock In Jakob Rehof, editor, {\em Proceedings of the 9th International
  Conference on Formal Structures for Computation and Deduction}, volume 299 of
  {\em Leibniz International Proceedings in Informatics}, pages 31:1--31:21,
  Dagstuhl, Germany, 2024. Schloss Dagstuhl -- Leibniz-Zentrum f{\"u}r
  Informatik.
\newblock \href {https://doi.org/10.4230/LIPIcs.FSCD.2024.31}
  {\path{doi:10.4230/LIPIcs.FSCD.2024.31}}.

\bibitem{ANS26FSCD-arxiv}
Takahito Aoto, Naoki Nishida, and Jonas Sch{\"o}pf.
\newblock Partial rewriting and value interpretation of logically constrained
  terms (full version).
\newblock {\em CoRR}, abs/2601.22191, 2026.
\newblock \href {https://doi.org/10.48550/arXiv.2601.22191}
  {\path{doi:10.48550/arXiv.2601.22191}}.

\bibitem{BN98}
Franz Baader and Tobias Nipkow.
\newblock {\em Term Rewriting and All That}.
\newblock Cambridge University Press, 1998.
\newblock \href {https://doi.org/10.1145/505863.505888}
  {\path{doi:10.1145/505863.505888}}.

\bibitem{FKN17tocl}
Carsten Fuhs, Cynthia Kop, and Naoki Nishida.
\newblock Verifying procedural programs via constrained rewriting induction.
\newblock {\em ACM Transactions on Computational Logic}, 18(2):14:1--14:50,
  2017.
\newblock \href {https://doi.org/10.1145/3060143} {\path{doi:10.1145/3060143}}.

\bibitem{HuetOppen}
G.~Huet and D.~C. Oppen.
\newblock Equations and rewrite rules: a survey.
\newblock Technical report, Stanford University, Stanford, CA, USA, 1980.

\bibitem{KN24jip}
Misaki Kojima and Naoki Nishida.
\newblock A sufficient condition of logically constrained term rewrite systems
  for decidability of all-path reachability problems with constant
  destinations.
\newblock {\em Journal of Information Processing}, 32:417--435, 2024.
\newblock \href {https://doi.org/10.2197/IPSJJIP.32.417}
  {\path{doi:10.2197/IPSJJIP.32.417}}.

\bibitem{Kop13}
Cynthia Kop.
\newblock Termination of {LCTRSs}.
\newblock In {\em Proceedings of the 13th Workshop on Termination}, pages 1--5,
  2013.
\newblock URL: \url{https://doi.org/10.48550/arXiv.1601.03206}.

\bibitem{KN13frocos}
Cynthia Kop and Naoki Nishida.
\newblock Term rewriting with logical constraints.
\newblock In Pascal Fontaine, Christophe Ringeissen, and Renate~A. Schmidt,
  editors, {\em Proceedings of the 9th International Symposium on Frontiers of
  Combining Systems}, volume 8152 of {\em Lecture Notes in Computer Science},
  pages 343--358, Berlin, Heidelberg, 2013. Springer Berlin Heidelberg.
\newblock \href {https://doi.org/10.1007/978-3-642-40885-4_24}
  {\path{doi:10.1007/978-3-642-40885-4_24}}.

\bibitem{KN14aplas}
Cynthia Kop and Naoki Nishida.
\newblock Automatic constrained rewriting induction towards verifying
  procedural programs.
\newblock In Jacques Garrigue, editor, {\em Proceedings of the 12th Asian
  Symposium on Programming Languages and Systems}, volume 8858 of {\em Lecture
  Notes in Computer Science}, pages 334--353. Springer, 2014.
\newblock \href {https://doi.org/10.1007/978-3-319-12736-1_18}
  {\path{doi:10.1007/978-3-319-12736-1_18}}.

\bibitem{KN15}
Cynthia Kop and Naoki Nishida.
\newblock {C}ons{T}rained {R}ewriting too{L}.
\newblock In Martin Davis, Ansgar Fehnker, Annabelle McIver, and Andrei
  Voronkov, editors, {\em Proceedings of the 20th International Conference on
  Logic for Programming, Artificial Intelligence, and Reasoning}, volume 9450
  of {\em Lecture Notes in Computer Science}, pages 549--557, Berlin,
  Heidelberg, 2015. Springer Berlin Heidelberg.
\newblock \href {https://doi.org/10.1007/978-3-662-48899-7_38}
  {\path{doi:10.1007/978-3-662-48899-7_38}}.

\bibitem{NW18}
Naoki Nishida and Sarah Winkler.
\newblock Loop detection by logically constrained term rewriting.
\newblock In Ruzica Piskac and Philipp R{\"u}mmer, editors, {\em Proceedings of
  the 10th International Conference on Verified Software. Theories, Tools, and
  Experiments}, volume 11294 of {\em Lecture Notes in Computer Science}, pages
  309--321, Cham, 2018. Springer International Publishing.
\newblock \href {https://doi.org/10.1007/978-3-030-03592-1_18}
  {\path{doi:10.1007/978-3-030-03592-1_18}}.

\bibitem{ReddyRI}
Uday~S. Reddy.
\newblock Term rewriting induction.
\newblock In Mark~E. Stickel, editor, {\em Proceedings of the 10th
  International Conference on Automated Deduction}, volume 449 of {\em Lecture
  Notes in Computer Science}, pages 162--177. Springer-Verlag, 1990.
\newblock \href {https://doi.org/10.1007/3-540-52885-7_86}
  {\path{doi:10.1007/3-540-52885-7_86}}.

\bibitem{RMM17}
Camilo Rocha, Jos{\'{e}} Meseguer, and C{\'{e}}sar~A. Mu{\~{n}}oz.
\newblock Rewriting modulo {SMT} and open system analysis.
\newblock {\em Journal of Logic and Algebraic Programming}, 86(1):269--297,
  2017.
\newblock \href {https://doi.org/10.1016/J.JLAMP.2016.10.001}
  {\path{doi:10.1016/J.JLAMP.2016.10.001}}.

\bibitem{SM23}
Jonas Sch{\"o}pf and Aart Middeldorp.
\newblock Confluence criteria for logically constrained rewrite systems.
\newblock In Brigitte Pientka and Cesare Tinelli, editors, {\em Proceedings of
  the 29th International Conference on Automated Deduction}, volume 14132 of
  {\em Lecture Notes in Artificial Intelligence}, pages 474--490, Cham, 2023.
  Springer Nature Switzerland.
\newblock \href {https://doi.org/10.1007/978-3-031-38499-8_27}
  {\path{doi:10.1007/978-3-031-38499-8_27}}.

\bibitem{SM25}
Jonas Sch{\"o}pf and Aart Middeldorp.
\newblock Automated analysis of logically constrained rewrite systems using
  crest.
\newblock In Arie Gurfinkel and Marijn Heule, editors, {\em Proceedings of the
  31st International Symposium on Principles and Practice of Declarative
  Programming}, volume 15696 of {\em Lecture Notes in Computer Science}, pages
  124--144, Cham, 2025. Springer Nature Switzerland.
\newblock \href {https://doi.org/10.1007/978-3-031-90643-5_7}
  {\path{doi:10.1007/978-3-031-90643-5_7}}.

\bibitem{SMM24}
Jonas Sch{\"o}pf, Fabian Mitterwallner, and Aart Middeldorp.
\newblock Confluence of logically constrained rewrite systems revisited.
\newblock In Christoph Benzmüller, Marijn~J.H. Heule, and Renate~A. Schmidt,
  editors, {\em Proceedings of the 12th International Joint Conference on
  Automated Reasoning}, volume 14740 of {\em Lecture Notes in Artificial
  Intelligence}, pages 298--316, Cham, 2024. Springer Nature Switzerland.
\newblock \href {https://doi.org/10.1007/978-3-031-63501-4_16}
  {\path{doi:10.1007/978-3-031-63501-4_16}}.

\bibitem{TSNA25LOPSTR}
Kanta Takahata, Jonas Sch{\"o}pf, Naoki Nishida, and Takahito Aoto.
\newblock Characterizing equivalence of logically constrained terms via
  existentially constrained terms.
\newblock In Santiago Escobar and Laura Titolo, editors, {\em Proceedings of
  the 35th International Symposium on Logic-Based Program Synthesis and
  Transformation}, volume 16117 of {\em Lecture Notes in Computer Science},
  pages 180--195. Springer, 2025.
\newblock \href {https://doi.org/10.1007/978-3-032-04848-6\_12}
  {\path{doi:10.1007/978-3-032-04848-6\_12}}.

\bibitem{TSNA25-PPDP}
Kanta Takahata, Jonas Sch{\"{o}}pf, Naoki Nishida, and Takahito Aoto.
\newblock Recovering commutation of logically constrained rewriting and
  equivalence transformations.
\newblock In Carlos~Olarte Ma{\l{}}gorzata~Biernacka, editor, {\em Proceedings
  of the 27th International Symposium on Principles and Practice of Declarative
  Programming}, Association for Computing Machinery, pages 9:1--9:13, 2025.
\newblock \href {https://doi.org/10.1145/3756907.3756916}
  {\path{doi:10.1145/3756907.3756916}}.

\bibitem{WM18}
Sarah Winkler and Aart Middeldorp.
\newblock Completion for logically constrained rewriting.
\newblock In H\'{e}l\`{e}ne Kirchner, editor, {\em Proceedings of the 3rd
  International Conference on Formal Structures for Computation and Deduction},
  volume 108 of {\em Leibniz International Proceedings in Informatics}, pages
  30:1--30:18, Dagstuhl, Germany, 2018. Schloss Dagstuhl -- Leibniz-Zentrum
  f{\"u}r Informatik.
\newblock \href {https://doi.org/10.4230/LIPIcs.FSCD.2018.30}
  {\path{doi:10.4230/LIPIcs.FSCD.2018.30}}.

\bibitem{WM21}
Sarah Winkler and Georg Moser.
\newblock Runtime complexity analysis of logically constrained rewriting.
\newblock In Maribel Fern{\'a}ndez, editor, {\em Proceedings of the 30th
  International Symposium on Logic-Based Program Synthesis and Transformation},
  volume 12561 of {\em Lecture Notes in Computer Science}, pages 37--55.
  Springer, 2021.
\newblock \href {https://doi.org/10.1007/978-3-030-68446-4_2}
  {\path{doi:10.1007/978-3-030-68446-4_2}}.

\end{thebibliography}

\newpage

\appendix

\section{Omitted Proofs}

\LemmaRedexIsAWeakRedex*
\begin{proof}
Suppose $\CTerm{X}{s}{\vec{x}}{\varphi}$ has a $\rho$-redex at position $p \in
\Pos(s)$ using substitution $\gamma$. Clearly, it suffices to show that
item~\Bfnum{4} in \Cref{def:weak-rho-redex} follows. Since
$\CTerm{X}{s}{\vec{x}}{\varphi}$ is satisfiable, $\ECO{\vec{x}}{\varphi}$ is
satisfiable. Thus, there exists a valuation $\rho$ such that $\vDash_{\xM,\rho}
\ECO{\vec{x}}{\varphi}$. By our assumption, $\xM \vDash_\rho
(\ECO{\vec{x}}{\varphi}) \Rightarrow (\ECO{\vec{z}}{\pi\gamma})$. Thus, 
$\vDash_{\xM,\rho} \ECO{\vec{z}}{\pi\gamma}$. Hence $\xM \vDash_\rho
(\ECO{\vec{x}}{\varphi}) \land (\ECO{\vec{z}}{\pi\gamma})$ from which it follows
that $(\ECO{\vec{x}}{\varphi}) \land (\ECO{\vec{z}}{\pi\gamma})$ is satisfiable.
\end{proof}

\TheoremMostGeneralRewriteStepImpliesWeakRewriteStep*
\begin{proof}
Suppose $\CTerm{X}{s}{\vec{x}}{\varphi} \to^p_{\rho,\gamma} \CTerm{Y}{t}{\vec{y}}{\psi}$.
Then $\CTerm{X}{s}{\vec{x}}{\varphi}$ has a $\rho$-redex at $p$ using $\gamma$,
which is a partial $\rho$-redex by \Cref{lem:redex is a weak redex}.
Hence,
we have a partial rewrite step
$\CTerm{X}{s}{\vec{x}}{\varphi} \leadsto^p_{\rho,\gamma} \CTerm{Y}{t}{\vec{y}}{\psi}$.
\end{proof}

\LemmaFromWeakRewriteStepToMostGeneralStep*
\begin{proof}
Suppose 
$\CTerm{X}{s}{\vec{x}}{\varphi} \leadsto^p_{\rho,\gamma} \CTerm{Y}{t}{\vec{y}}{\psi}$.
Then, $\CTerm{X}{s}{\vec{x}}{\varphi}$ has a partial $\rho$-redex at $p$ using $\gamma$.
Thus, we have
(1) $\Dom(\gamma) = \Var(\ell)$,
(2) $s|_p = \ell\gamma$,
(3) $\gamma(x) \in \Val \cup X$ for all $x \in \Var(\ell) \cap Z$, and
(4) $(\ECO{\vec{x}}{\varphi}) \land (\ECO{\vec{z}}{\pi\gamma})$ is satisfiable in $\xM$.

By $\SET{\vec{z}} = \Var(\pi) \setminus \Var(\ell)$ and condition (1)
we have $\FVar(\ECO{\vec{z}}{\pi\gamma}) \subseteq \SET{ \gamma(x) \mid x \in \Var(\ell) \cap \Var(\pi)}$.
Then, since $\Var(\pi) \subseteq Z$ by the definition of rewrite rule,
it follows that~(4.5) $\FVar(\ECO{\vec{z}}{\pi\gamma}) \subseteq X$ by condition (3).

By $\Var(\rho) \cap \Var(s,\varphi) = \varnothing$,
we have $\SET{\vec{x}} \cap \SET{\vec{z}} = \varnothing$
and $\Var(\varphi) \cap \SET{\vec{z}} = \varnothing$.
Furthermore, by $\Var(\pi\gamma) \subseteq \Var(s) \cup \SET{\vec{z}}$,
we  have $\FVar(\exists \vec{z}.\pi\gamma) \cap \SET{\vec{x}} = \varnothing$.
Hence, $\ECO{\pvec{x},\vec{z}}{\varphi \land \pi\gamma}$
is a prenex normal form of $(\ECO{\pvec{x}}{\varphi}) \land (\ECO{\vec{z}}{\pi\gamma})$.
It follows that
(5) $\FVar(\ECO{\vec{x},\vec{z}}{\varphi \land \pi\gamma}) 
= \FVar((\ECO{\vec{x}}{\varphi}) \land (\ECO{\vec{z}}{\pi\gamma}))$
and
(6) $\vDash_\xM 
(\ECO{\vec{x},\vec{z}}{\varphi \land \pi\gamma}) 
\Leftrightarrow
((\ECO{\vec{x}}{\varphi}) \land (\ECO{\vec{z}}{\pi\gamma}))$.

By the definition  of constrained terms,
we have $\FVar(\ECO{\vec{x}}{\varphi}) \subseteq X$.
Hence, it follows from~(4.5)
and~(5)
that
(7) $\FVar(\ECO{\vec{x},\vec{z}}{\varphi \land \pi\gamma}) \subseteq X$.

By condition (1), we have $\SET{\vec{z}} \subseteq \Var(\pi\gamma)$.
Then, since $\SET{\vec{x}} \subseteq \Var(\varphi)$,
we have $\SET{\vec{x},\vec{z}} \subseteq \Var(\varphi \land \pi\gamma)$.
Also by $\SET{\vec{x}} \cap \Var(s) = \varnothing$
and $\Var(\rho) \cap \Var(s,\varphi) = \varnothing$,
we have $\SET{\vec{x},\vec{z}} \cap \Var(s) = \varnothing$.
Together with (7) we have that, 
$\CTerm{X}{s}{\vec{x},\vec{z}}{\varphi \land \pi\gamma}$
is an existentially constrained term.

From (6),
we have $\vDash_\xM (\ECO{\vec{x},\vec{z}}{\varphi \land \pi\gamma}) \Rightarrow (\ECO{\vec{x}}{\varphi})$
and $\vDash_\xM (\ECO{\vec{x},\vec{z}}{\varphi \land \pi\gamma}) \Rightarrow (\ECO{\vec{z}}{\pi\gamma})$.
The former gives
$\CTerm{X}{s}{\vec{x},\vec{z}}{\varphi \land \pi\gamma} \subsetsim \CTerm{X}{s}{\vec{x}}{\varphi}$.

Consider a renamed variant $\rho': \CRu{Z'}{\ell'}{r'}{\pi'}$ of $\rho$
satisfying $\Var(\rho') \cap \Var(s,\varphi \land \pi\gamma) = \varnothing$.
Let $\delta$ be a renaming such that $\rho' = \rho \delta$,
and let $\gamma'=  \gamma \circ \delta^{-1}$.
Then, we have $\ell'\gamma' = \ell\gamma$,
$r'\gamma' = r\gamma$ and $\pi'\gamma' = \pi\gamma$.
It also follows that 
$\vDash_{\xM} (\ECO{\vec{x},\vec{z}}{\varphi \land \pi\gamma}) \Rightarrow (\ECO{\pvec{z}'}{\pi'\gamma'})$,
where $\SET{\pvec{z}'} = \Var(\pi') \setminus \Var(\ell')$.
We obtain that 
$\CTerm{X}{s}{\vec{x},\vec{z}}{\varphi \land \pi\gamma}$ 
has a $\rho'$-redex at $p$ using $\gamma'$,
and 
$\CTerm{X}{s}{\vec{x},\vec{z}}{\varphi \land \pi\gamma} \to_\rho 
\CTerm{Y'}{t'}{\pvec{y}'}{\psi'}$,
where 
$t' = s[r'\gamma']_p$,
$\psi' = \varphi \land \pi\gamma \land \pi'\gamma'$,
$\SET{\vec{y}'} = \Var(\psi') \setminus \Var(t')$, and
$Y' = \ExVar(\rho) \cup (X \cap \Var(t'))$.
Observe now that $t' = s[r'\gamma']_p = s[r\gamma]_p = t$,
that $\vDash_{\xM} \psi \Leftrightarrow \psi'$ 
by $\psi = \varphi \land \pi\gamma$ and
$\psi' = \varphi \land \pi\gamma \land \pi'\gamma'
       = \varphi \land \pi\gamma \land \pi\gamma$,
and that $\Var(\psi') = \Var(\psi)$.
Hence, 
$\CTerm{Y'}{t'}{\pvec{y}'}{\psi'} \sim \CTerm{Y}{t}{\vec{y}}{\psi}$.
\end{proof}

\LemmaSubsumptionByInstances*
\begin{proof}
Clearly~\Bfnum{2} immediately follows from~\Bfnum{1} by the definition of equivalence.
To show~\Bfnum{1}, let $\CTerm{X}{s}{\vec{x}}{\varphi}$, $\CTerm{Y}{t}{\vec{y}}{\psi}$
be existentially constrained terms.
We start by proving the \emph{only if}-direction.
Assume that $\CTerm{X}{s}{\vec{x}}{\varphi} \subsetsim \CTerm{Y}{t}{\vec{y}}{\psi}$.
To show 
$\AllInst{\CTerm{X}{s}{\vec{x}}{\varphi}} \subseteq \AllInst{\CTerm{Y}{t}{\vec{y}}{\psi}}$,
take an element $u \in \AllInst{\CTerm{X}{s}{\vec{x}}{\varphi}}$.
Then, by definition, there exists an $X$-valued substitution $\sigma$ with 
$\sigma \vDash_{\xM} \ECO{\vec{x}}{\varphi}$
and $u = s\sigma$. 
Then, by definition of $\subsetsim$, there exists a $Y$-valued substitution
$\gamma$ with $\gamma \vDash_\xM \ECO{\vec{y}}{\psi}$ and $s\sigma = t\gamma$.
It follows that $u = s\sigma = t\gamma \in \AllInst{\CTerm{Y}{t}{\vec{y}}{\psi}}$.
For the \emph{if}-direction we assume that
$\AllInst{\CTerm{X}{s}{\vec{x}}{\varphi}} \subseteq \AllInst{\CTerm{Y}{t}{\vec{y}}{\psi}}$ holds.
Fix an arbitrary $X$-valued substitution $\sigma$ with 
$\sigma \vDash_\xM \ECO{\vec{x}}{\varphi}$.
Then, as we have $s\sigma \in \AllInst{\CTerm{X}{s}{\vec{x}}{\varphi}}$,
by our assumption
$s\sigma \in \AllInst{\CTerm{Y}{t}{\vec{y}}{\psi}}
= \SET{t\gamma \mid \gamma(Y) \subseteq \Val, \gamma \vDash_\xM \ECO{\vec{y}}{\psi}}$.
Thus, $s\sigma = t\gamma$ for some $Y$-valued substitution $\gamma$
such that $\gamma \vDash_\xM \ECO{\vec{y}}{\psi}$.
Hence, we conclude $\CTerm{X}{s}{\vec{x}}{\varphi} \subsetsim \CTerm{Y}{t}{\vec{y}}{\psi}$.
\end{proof}

\LemmaCharacterizationOfWeakRewriteSteps*
\begin{proof}
Let $\rho: \CRu{Z}{\ell}{r}{\pi}$ be a left-linear constrained rewrite rule
satisfying $\Var(\rho) \cap \Var(s,\varphi) = \varnothing$ together with
$\CTerm{X}{s}{\vec{x}}{\varphi} \leadsto^p_{\rho,\gamma} \CTerm{Y}{t}{\vec{y}}{\psi}$.
We have 
(1) $\Dom(\gamma) = \Var(\ell)$,
(2) $s|_p = \ell\gamma$,
(3) $\gamma(x) \in \Val \cup X$ for all $x \in \Var(\ell) \cap Z$, and
(4) $(\ECO{\vec{x}}{\varphi}) \land (\ECO{\vec{z}}{\pi\gamma})$ is satisfiable
where $\SET{\vec{z}} = \Var(\pi) \setminus \Var(\ell)$.
Furthermore, we have
(5) $t = s[r\gamma]_p$,
(6) $\psi = \varphi \land \pi\gamma$,
(7) $\SET{\vec{y}} = \Var(\psi) \setminus \Var(t)$, and
(8) $Y = \ExVar(\rho) \cup (X \cap \Var(t))$.

\begin{enumerate}
\item 
By~(4),
there exists a valuation $\delta$
such that $\vDash_\delta (\ECO{\vec{x}}{\varphi}) \land (\ECO{\vec{z}}{\pi\gamma})$.
Then, $\vDash_\delta \ECO{\vec{z}}{\pi\gamma}$,
and hence $\vDash_{\delta[\vec{z} \mapsto \vec{w}]} \pi\gamma$ 
for some $\vec{w} \in \Val^*$.
Since $ \SET{ \vec{z} } \cap X = \varnothing$,
let $\hat \delta$ be a substitution
obtained from $\delta[\vec{z} \mapsto \vec{w}]$ by restricting its domain
to $X \cup \Var(\pi\gamma) \cup \ExVar(\rho)$.
Note that, by $\FVar(\ECO{\vec{x}}{\varphi}) \subseteq  X$,
we have 
$\FVar(\ECO{\vec{x}}{\varphi}) \subseteq  \VDom(\hat \delta)$
and 
$\vDash_{\hat \delta} \ECO{\vec{x}}{\varphi}$,
and thus $\hat\delta \vDash \ECO{\vec{x}}{\varphi}$ follows.
Moreover, by choosing this specific $\hat \delta$,
it follows that $\hat \delta(X) \subseteq \Val$.
Thus, $s \hat\delta\in \AllInst{\CTerm{X}{s}{\vec{x}}{\varphi}}$.

Let $u = s\hat \delta$ and let us 
show that 
$u \to^p_{\AllInst{\rho}} v$
for some $v$.
For this, it suffices to show
$\ell \gamma \hat\delta \to^\epsilon_{\AllInst{\rho}} r \gamma\hat\delta$
as $u|_p = (s \hat\delta)|_p
= s|_p \hat\delta
= \ell \gamma \hat\delta$.
Let 
$\sigma = (\hat\delta \circ \gamma)|_{Z \cap \Var(\ell,r)}$.
It remains to show that
$\Dom(\sigma) = \VDom(\sigma) = Z \cap \Var(\ell,r)$
and
$\sigma \vDash \ECO{\pvec{z}'}{\pi}$
where $\SET{ \pvec{z}' } = \Var(\pi) \setminus \Var(\ell,r)$.

For the former, by our choice of $\sigma$, it remains to
show $Z \cap \Var(\ell,r) \subseteq \VDom(\sigma)$.
Assume $z \in Z \cap \Var(\ell,r)$.
We distinguish two cases.
If $z \in Z \cap \Var(\ell)$,
then we have $\gamma(z) \in X$ by~(3);
thus, $\sigma(z) = \hat\delta(\gamma(z)) \in \Val$ by
our choice of $\hat\delta$.
Otherwise,  we have $z \in Z \cap \ExVar(\rho) =\ExVar(\rho)$.
Then, $\sigma(z) = \hat\delta(\gamma(z)) = \hat\delta(z) \in \Val$
by $\ExVar(\rho) \cap \Dom(\gamma) = \varnothing$
and our choice of $\hat\delta$.

To show $\FVar(\ECO{\pvec{z}'}{\pi}) \subseteq \VDom(\sigma)$,
let $z \in \FVar(\ECO{\pvec{z}'}{\pi})$,
i.e., 
$z \in \Var(\pi) \cap \Var(\ell,r)$.
If $z \in \Var(\ell)$, then $z \in \Var(\ell) \cap X$ by $z \in \Var(\pi)$,
and hence $\gamma(z) \in \Val$. Thus, $\sigma(z) = \hat\delta(\gamma(z)) \in \Val$ follows.
Otherwise $z \in \ExVar(\rho)$.
Then $\gamma(z) = z$, and hence $\sigma(z) = \hat\delta(z) \in \Val$  by our choice of $\hat\delta$.
Thus, $\FVar(\ECO{\pvec{z}'}{\pi}) \subseteq \VDom(\sigma)$
and it remains to show $\vDash_{\sigma} \ECO{\pvec{z}'}{\pi}$.

By our choice of $\hat\delta$,
we have $\vDash_{\hat \delta} \pi\gamma$.
Thus, $\vDash_{[\![ \gamma]\!]_{\hat \delta}} \pi$.
Hence $\vDash_{[\![ \gamma]\!]_{\hat \delta}} \ECO{\pvec{z}'}{\pi}$.
We now show $[\![ \gamma]\!]_{\hat \delta}(x) = 
\sigma(x)$ for all $x \in \FVar(\ECO{\pvec{z}'}{\pi})$.
From $x \in \FVar(\ECO{\pvec{z}'}{\pi})$
it follows $x \in \Var(\pi) \setminus \Var(\ell,r)$,
and hence $x \notin \Var(\ell)$.
Thus, $\gamma(x) = x$.
Hence 
$[\![ \gamma(x) ]\!]_{\hat \delta}
= [\![ x ]\!]_{\hat \delta}
= \hat\delta(x)
= \hat\delta(\gamma(x))
= \sigma(x)$ follows.

Hence we conclude the claim.

\item 
Suppose 
$w \in \AllInst{\CTerm{Y}{t}{\vec{y}}{\psi}}$.
Then, by definition,
(9) $w = t\delta$, (10) $\delta(Y) \subseteq \Val$, and 
(11) $\delta \vDash_\xM \ECO{\vec{y}}{\psi}$ for some substitution $\delta$.
without loss of generality assume $\Dom(\delta) \subseteq \Var(t)$.

As $\Var(t) \cap \SET{\vec{y}} = \varnothing$,
one can extend $\delta$ to $\delta' = \delta \cup \SET{ \vec{y} \mapsto \vec{v} }$ so that
$\vDash_{\delta'} \psi$ for some $\vec{v} \in \Val^*$.
We further extend 
$\delta'$ to $\delta'' = \delta' \cup \SET{ \pvec{x} \mapsto \vec{w} }$
where $\SET{\pvec{x}} = X \setminus \Var(t)$
by taking an arbitrary $\vec{w} \in \Val^*$.

We now show $s\delta'' \in \AllInst{\CTerm{X}{s}{\vec{x}}{\varphi}}$.
Firstly, by $Y = (X \cap \Var(t)) \cup (\Var(r) \setminus \Var(\ell))$ and~(10), 
we have $X \cap \Var(t) \subseteq \VDom(\delta)$.
Thus, $\delta''(X) \subseteq \Val$, by our choice of $\delta''$.
It remains to show 
$\delta'' \vDash \ECO{\vec{x}}{\varphi}$.
We have $\FVar(\ECO{\vec{y}}{\psi}) \subseteq \VDom(\delta)$ by (11).
Hence, 
$\FVar(\ECO{\vec{x}}{\varphi}) 
\subseteq \Var(\varphi) 
\subseteq \Var(\psi) 
\subseteq \VDom(\delta')
\subseteq \VDom(\delta'')$.
Also, by $\vDash_{\delta'} \psi~(= \varphi \land \pi\gamma)$,
we have $ \vDash_{\delta'} \ECO{\vec{x}}{\varphi}$.
Thus, since $\delta'(x) = \delta''(x)$ for $x \notin \SET{ \vec{x} }$,
we have $\vDash_{\delta''} \ECO{\vec{x}}{\varphi}$.
Hence $\delta'' \vDash \ECO{\vec{x}}{\varphi}$ as claimed.
Thus, $s\delta'' \in \AllInst{\CTerm{X}{s}{\vec{x}}{\varphi}}$.

Let $u = s\delta''$.
Since $\delta(x) = \delta''(x)$ for all $x \in \Var(t)$,
we have $u|_p = s|_p\delta'' = \ell\gamma\delta''$
and $w = t\delta = t\delta'' 
= s[r\gamma]_p\delta'' 
= s\delta'' [r\gamma\delta'' ]_p
= u[r\gamma\delta'' ]_p$.
Thus, to obtain $u \to^p_{\AllInst{\rho}} w$,
it remains to show 
$\Dom(\sigma) = \VDom(\sigma) = Z \cap \Var(\ell,r)$
and
$\sigma \vDash \ECO{\pvec{z}'}{\pi}$, 
where $\sigma$ is a substitution with 
$\Dom(\sigma) = Z \cap \Var(\ell,r)$
given by $\sigma(x) = \delta''(\gamma(x))$.

For the former, it suffices to 
show $Z \cap \Var(\ell,r) \subseteq \VDom(\sigma)$
by our choice of $\sigma$.
Let $z \in Z \cap \Var(\ell,r)$.
If $z \in Z \cap \Var(\ell)$ then $\gamma(z) \in \Val \cup X$.
Then, by $\delta''(X) \subseteq \Val$,
it follows $\sigma(z) = \delta''(\gamma(z)) \in \Val$.
If $z \in Z \cap \ExVar(\rho)$,
then $z \in Y$ by the definition of rewrite steps,
and hence $\sigma(z) = \delta''(\gamma(z)) = \delta(z)  \in \Val$ by (10).
Hence, all in all,
$Z \cap \Var(\ell,r) \subseteq \VDom(\sigma)$ holds.

As $\FVar(\ECO{\pvec{z}'}{\pi}) \subseteq Z \cap \Var(\ell,r)$,
we have $\FVar(\ECO{\pvec{z}'}{\pi}) \subseteq \VDom(\sigma)$.
Furthermore,
from $\vDash_{\delta'} \psi~(= \varphi \land \pi\gamma)$
and $\Var(\pi\gamma) \cap \SET{ \vec{x} } = \varnothing$,
it follows that $\vDash_{\delta''} \pi\gamma$.
Hence,
we have $\vDash_{[\![ \gamma]\!]_{\delta''}} \pi$.
By $\Dom(\gamma) \cap \Var(\pi) \subseteq \Var(\ell) \cap X$,
we have 
$[\![ \gamma]\!]_{\delta''}(x) 
= \delta''(\gamma(x))
= \sigma(x)$ for all $x \in \Var(\pi)$.
Hence, $\vDash_{\sigma} \pi$.
Thus,  $\vDash_{\sigma} \ECO{\pvec{z}'}{\pi}$, 
and $\sigma \vDash \ECO{\pvec{z}'}{\pi}$ follows by $\FVar(\ECO{\pvec{z}'}{\pi}) \subseteq \VDom(\sigma)$.
This concludes the claim.

\item 
Suppose $u \in \AllInst{\CTerm{X}{s}{\vec{x}}{\varphi}}$
and 
$u \to^p_{\AllInst{\rho}} v$.
Then $u = s\delta$ 
for substitution $\delta$
such that $\delta(X) \subseteq \Val$ and $\delta \vDash \ECO{\vec{x}}{\varphi}$.
without loss of generality suppose $\Dom(\delta) \subseteq \Var(s)$.

From 
$u \to^p_{\AllInst{\rho}} v$,
there exists a substitution $\sigma$
such that $u|_p = \ell\sigma$, 
$v = u[r\sigma]_p$,
$Z \cap \Var(\ell,r) \subseteq \VDom(\sigma)$
and $\sigma \vDash \ECO{\pvec{z}'}{\pi}$
where $\SET{\pvec{z}'} = \Var(\pi) \setminus \Var(\ell,r)$.
Since $u|_p = s|_p\delta = \ell\gamma\delta$,
we have 
$\sigma(x) = \delta(\gamma(x))$ for all $x \in \Var(\ell)$.
We are now going to show
$v = u[r\sigma]_p \in \AllInst{\CTerm{Y}{t}{\vec{y}}{\psi}}$.

By $Z \cap \Var(\ell,r) \subseteq \VDom(\sigma)$,
we know $\ExVar(\rho) \subseteq \VDom(\sigma)$.
Thus, for $x \in \ExVar(\rho)$,
we have $\sigma(x) \in \Val$.
Thus, if $x = r|_{q_x}$ then $\sigma(x) = (r\sigma)|_{q_x} = v|_{p.q_x}$ holds.
Thus, one can extend $\delta$ to $\delta'$ by $\delta'(x) = v|_{p.q_x}$ for
$x \in \ExVar(\rho)$.
Then, for $x \in \ExVar(\rho)$, we have
$\sigma(x) = v|_{p.q_x} = \delta'(x) = \delta'(\gamma(x))$
and also clearly, $\delta(x) = \delta'(x)$ holds also for $x \in \Var(\ell)$.
Thus, we have
$\sigma(x) = \delta'(\gamma(x))$ for all $x \in \Var(\ell,r)$.

Note that $\Var(s) \cap \ExVar(\rho) = \varnothing$,
we have $s \delta = s\delta'$.
Thus, 
$u[r\sigma]_p = 
s\delta[r\sigma]_p =
s\delta'[r\gamma\delta']_p =
(s[r\gamma]_p)\delta' = t\delta'$.

Thus, it remains to to show 
$\delta'(Y)  \subseteq \Val$
and
$\delta' \vDash \ECO{\vec{y}}{\psi}$,
i.e.,
$\delta' \vDash \ECO{\vec{y}}{\varphi \land \pi\gamma}$.

Let $y \in Y = (X \cap \Var(t)) \cup \ExVar(\rho)$.
If $y \in X \cap \Var(t) \subseteq X$
then $\delta'(y) = \delta(y) \in \Val$ by $\delta(X) \subseteq \Val$.
If $y \in \ExVar(\rho)$, then 
$\delta'(y)  \in \Val$ by our construction of $\delta'$.
Thus, $\delta'(Y)  \subseteq \Val$ holds.

As $\FVar(\ECO{\vec{y}}{\psi}) \subseteq Y$,
$\FVar(\ECO{\vec{y}}{\psi}) \subseteq \VDom(\delta')$ follows
from $\delta'(Y)  \subseteq \Val$.
Thus, it remains to show $\vDash_{\delta'} \ECO{\vec{y}}{\psi}$.
From $\vDash_\delta \ECO{\vec{x}}{\varphi}$
we have $\vDash_{\delta'} \ECO{\vec{x}}{\varphi}$.

From $\sigma \vDash \ECO{\pvec{z}'}{\pi}$,
we have $\vDash_{\delta' \circ \gamma} \ECO{\pvec{z}'}{\pi}$.
Thus, $\vDash_{\delta'} (\ECO{\pvec{z}'}{\pi})\gamma$.
As $\Dom(\gamma) = \Var(\ell)$ and $\SET{\pvec{z}'} = \Var(\pi) \setminus \Var(\ell,r)$,
and $\Dom(\gamma) \cap \SET{\pvec{z}'} = \varnothing$.
Also, by $\gamma(\Var(\ell)) \subseteq \Var(s)$,
and $\Ran(\gamma) \cap \SET{\pvec{z}'} = \varnothing$.
Thus, $\vDash_{\delta'} \ECO{\pvec{z}'}{\pi\gamma}$.
Hence $\vDash_{\delta'} (\ECO{\vec{x}}{\varphi}) \land (\ECO{\pvec{z}'}{\pi\gamma})$. 
From this, 
it follows $\vDash_{\delta'} \ECO{\vec{x},\pvec{z}'}{\varphi \land \pi\gamma}$. 
Since $\SET{\vec{x}} \cap \Var(t) = \varnothing$ and 
$\SET{\pvec{z}'} \cap \Var(t) = \varnothing$,
it follows $\vDash_{\delta'} \ECO{\vec{y}}{\varphi \land \pi\gamma}$.
\end{enumerate}
\end{proof}

\TheoremCharacterizationbyInterpreationForWeakReduction*
\begin{proof}
\Bfnum{1} follows from item~\Bfnum{2} of \Cref{lem:Characterization of Weak Rewrite Steps}.
\Bfnum{2} follows from the items~\Bfnum{1} and~\Bfnum{3} of \Cref{lem:Characterization of Weak Rewrite Steps}.
%
\end{proof}

\LemmaCharacterizationOfMGRewriteSteps*
\begin{proof}
Let $\rho: \CRu{Z}{\ell}{r}{\pi}$ be a left-linear constrained rewrite rule
satisfying $\Var(\rho) \cap \Var(s,\varphi) = \varnothing$,
and suppose $\CTerm{X}{s}{\vec{x}}{\varphi} \to^p_{\rho,\gamma} \CTerm{Y}{t}{\vec{y}}{\psi}$.
Thus, we have 
(1) $\Dom(\gamma) = \Var(\ell)$,
(2) $s|_p = \ell\gamma$,
(3) $\gamma(x) \in \Val \cup X$ for all $x \in \Var(\ell) \cap Z$, and
(4) $\vDash_\xM (\ECO{\vec{x}}{\varphi}) \Rightarrow (\ECO{\vec{z}}{\pi\gamma})$,
where $\SET{\vec{z}} = \Var(\pi) \setminus \Var(\ell)$.
Furthermore, we have
(5) $t = s[r\gamma]_p$,
(6) $\psi = \varphi \land \pi\gamma$,
(7) $\SET{\vec{y}} = \Var(\psi) \setminus \Var(t)$, and
(8) $Y = \ExVar(\rho) \cup (X \cap \Var(t))$.

Let $u \in \AllInst{\CTerm{X}{s}{\vec{x}}{\varphi}}$.
Then, by definition,
(9) $u = s\delta$, (10) $\delta(X) \subseteq \Val$, and 
(11) $\delta \vDash_\xM \ECO{\vec{x}}{\varphi}$ for some substitution $\delta$
such that $\Dom(\delta) \subseteq \Var(s)$.

Now we define a substitution $\sigma$ as follows:
\begin{itemize}
\item For $x \in \Var(\ell)$, let $\sigma(x) = \delta(\gamma(x))$.

\item For $x \in \Var(r) \setminus (\ell)$,
we define $\sigma(x)$ as follows.
For this, let $\SET{\vec{w}} = \Var(r) \setminus \Var(\ell)$.
Let $\hat \delta$ be a valuation such that
$\hat \delta(x) = \delta(x)$ for all $x \in X$.
Then, by (11) and $\FVar(\ECO{\vec{x}}{\varphi}) \subseteq X$, we have 
$\vDash_{\hat \delta} \ECO{\vec{x}}{\varphi}$.
Hence, by (4), we know
$\vDash_{\hat \delta} (\ECO{\vec{z}}{\pi\gamma})$.
From this, it follows 
$\vDash_{\hat \delta [\vec{w} \mapsto \vec{v}]} \ECO{\pvec{z}'}{\pi\gamma}$
for some $\vec{v} \in \Val^*$,
where $\SET{\pvec{z}'} = \Var(\pi) \setminus \Var(\ell,r)$
and
$\hat \delta [\vec{w} \mapsto \vec{v}]$ denotes the update of 
the valuation $\hat \delta$ by  $\vec{w} \mapsto \vec{v}$.
Then, we take $\sigma(\vec{w}) = \vec{v}$.

\item For $x \in \xV \setminus \Var(\ell,r)$,
we take an arbitrary $\sigma(x)$.
\end{itemize}

Here, for the later use, let 
$\hat\delta' = \hat \delta[\vec{w} \mapsto \vec{v}]$.
Clearly, since $\FVar(\ECO{\vec{x}}{\varphi}) \cap \SET{\vec{w}} = \varnothing$,
we have 
$\vDash_{\hat \delta'} \ECO{\vec{x}}{\varphi}$.
and hence 
(12) $\vDash_{\hat \delta'} (\ECO{\vec{x}}{\varphi}) \land (\ECO{\pvec{z}'}{\pi\gamma})$ holds.

From~(2) and~(9),
it follows that 
$u|_p = s\delta|_p = s|_p\delta = \ell\gamma\delta = \ell\sigma$ by 
our choice of $\sigma$ for $x \in \Var(\ell)$.
To obtain $v$ such that 
$u \to^p_{\AllInst{\rho}} v$,
it suffices to show 
$\ell\sigma \to_{\AllInst{\rho}} r\sigma$.
For this, 
we take $\sigma' = \sigma|_{Z \cap \Var(\ell,r)}$
and show 
$\Dom(\sigma') = \VDom(\sigma') = Z \cap \Var(\ell,r)$
and 
$\sigma' \vDash \ECO{\pvec{z}'}{\pi}$.


The former follows from $Z \cap \Var(\ell,r)  \subseteq \VDom(\sigma')$
by our choice of $\sigma'$.
Let $z \in Z \cap \Var(\ell,r)$.
If $z \in Z \cap \Var(\ell)$ then $\gamma(z) \in \Val \cup X$ by (3),
and thus, $\sigma'(z) = \delta(\gamma(z)) \in \Val$ by (10).
Otherwise 
$z \in (Z \cap \Var(r)) \setminus (Z \cap \Var(\ell))
\subseteq \Var(r) \setminus \Var(\ell)$,
and hence $\sigma'(z) \in \Val$ by our choice of $\sigma'$
for $z \in \Var(r) \setminus \Var(\ell)$.
Thus, $Z \cap \Var(\ell,r)  \subseteq \VDom(\sigma')$.

We proceed to show the latter.
We first show $\FVar(\ECO{\pvec{z}'}{\pi}) \subseteq \VDom(\sigma')$.
This follows from the former
because $\FVar(\ECO{\pvec{z}'}{\pi}) =  
\Var(\pi) \setminus (\Var(\pi) \setminus \Var(\ell,r))
= \Var(\pi) \cap \Var(\ell,r)
\subseteq  Z \cap \Var(\ell,r)$.

It remains to show 
$\vDash_{\hat \sigma} \ECO{\pvec{z}'}{\pi}$
for the valuation $\hat \sigma$ such that 
$\hat \sigma(z) = \sigma'(z)$ for all $z \in \FVar(\ECO{\pvec{z}'}{\pi})$.
From $\Ran(\gamma) \subseteq \Var(s)$,
$\Ran(\gamma) \cap \SET{\pvec{z}'} = \varnothing$.
We also have 
$\Dom(\gamma) = \Var(\ell)$ and $\Var(\ell) \cap \SET{\pvec{z}'} = \varnothing$,
and hence 
$\Dom(\gamma) \cap \SET{\pvec{z}'} = \varnothing$.
Thus, it follows from 
$\vDash_{\hat\delta'} \ECO{\pvec{z}'}{\pi\gamma}$
that 
$\vDash_{\hat\delta'} (\ECO{\pvec{z}'}{\pi})\gamma$.
Hence, we have $\vDash_{[\![ \gamma ]\!]_{\hat\delta'}} \ECO{\pvec{z}'}{\pi}$.

We now show $[\![ \gamma ]\!]_{\hat\delta'}
(z) = \hat \sigma(z)$ for all 
$z \in \FVar(\ECO{\pvec{z}'}{\pi}) = \Var(\pi) \cap \Var(\ell,r)$.
From this, our claim $\vDash_{\hat \sigma} \ECO{\pvec{z}'}{\pi}$ clearly follows.
We distinguish two cases.
\begin{itemize}
\item 
Suppose $z \in \Var(\pi) \cap \Var(\ell)$.
Then,
$[\![ \gamma ]\!]_{\hat\delta'}(z) 
= [\![ \gamma ]\!]_{\hat \delta[\vec{w} \mapsto \vec{v}]}(z) 
= [\![ \gamma(z) ]\!]_{\hat \delta[\vec{w} \mapsto \vec{v}]}
= [\![ \gamma(z) ]\!]_{\hat \delta}
= \hat \delta(\gamma(z))
= \delta(\gamma(z)) = \hat{\sigma}(z)$,
as $\gamma(z) \in X \cup \Val$.

\item 
Otherwise, let $z \in \Var(\pi) \cap (\Var(r) \setminus \Var(\ell))$.
Then, we have
$[\![ \gamma ]\!]_{\hat\delta'}(z) 
= [\![ \gamma ]\!]_{\hat \delta[\vec{w} \mapsto \vec{v}]}(z) 
= [\![ \gamma(z) ]\!]_{\hat \delta[\vec{w} \mapsto \vec{v}]}
= [\![ z ]\!]_{\hat \delta[\vec{w} \mapsto \vec{v}]}
= \SET{ \vec{w} \mapsto \vec{v} }(z)
= \hat{\sigma}(z)$.
\end{itemize}
%
This concludes the claim.
\end{proof}


\LemmaObtainAllInstFromValueInst*
\begin{proof}
($\subseteq$)
Suppose $u \in \AllInst{\CTerm{X}{s}{\vec{x}}{\varphi}}$.
Then there exists a substitution $\sigma$
such that $u = s\sigma$, $\sigma(X) \subseteq \Val$
and $\sigma \vDash \ECO{\vec{x}}{\varphi}$.
Take $\sigma' := \sigma|_X$ and $\theta := \sigma|_{\overline{X}}$.
Then, 
(1) $\sigma'(X) = \sigma(X) \subseteq \Val$,
(2) since $\sigma'(x) = x \in \xV$ for $x \in \Var(s) \setminus X$,
$\sigma'(\Var(s) \setminus X) \subseteq \xV$ holds, and
(3)
since $\FVar(\ECO{\vec{x}}{\varphi}) \subseteq X$
and $\sigma'(x) = \sigma(x)$ for any $x \in X$,
$\sigma' \vDash \ECO{\vec{x}}{\varphi}$ holds.
Thus,  it follows that 
$s\sigma' \in \AllInst{\CTerm{X}{s}{\vec{x}}{\varphi}}_\mathsf{v}$.
Finally, since $\sigma = \theta \circ \sigma'$,
we obtain $u = s\sigma = s\sigma'\theta \in 
\SET{ t\theta \mid t \in \AllInst{\CTerm{X}{s}{\vec{x}}{\varphi}}_\mathsf{v}$,
$\theta$ is a substitution$}$.
($\supseteq$)
Suppose     
$u = t\theta$ for some substitution $\theta$ and
$t \in \AllInst{\CTerm{X}{s}{\vec{x}}{\varphi}}_\mathsf{v}$.
Then, 
$t = s\sigma$ for some $\sigma$
such that (1) $\sigma(X) \subseteq \Val$,
(2) $\sigma(\Var(s) \setminus X) \subseteq \xV$,
(3) $\sigma(x) \neq \sigma(y)$ for any $x,y \in \Var(s) \setminus X$
with $x \neq y$, and
(4) $\sigma \vDash \ECO{\vec{x}}{\varphi}$.
Take $\sigma' := \theta \circ \sigma$.
Then, $\sigma'(x) = \theta(\sigma(x)) = \sigma(x) \in \Val$ for any $x \in X$.
Thus, 
$\sigma'(X) = \sigma(X) \subseteq \Val$.
Also, $\sigma' \vDash \ECO{\vec{x}}{\varphi}$ follows
by $\sigma \vDash \ECO{\vec{x}}{\varphi}$
and $\FVar(\ECO{\vec{x}}{\varphi}) \subseteq X$.
Thus, $s\sigma' \in \AllInst{\CTerm{X}{s}{\vec{x}}{\varphi}}$.
Finally, the claim follows as $u = t\theta = s\sigma\theta = s\sigma'$.
\end{proof}

\LemmaObtainValueInstFromAllInst*
\begin{proof}
The claim follows from 
(1) 
$\AllInst{\CTerm{X}{s}{\vec{x}}{\varphi}}_\mathsf{v}
\subseteq \AllInst{\CTerm{X}{s}{\vec{x}}{\varphi}}$, and
(2) 
if $u,v \in \AllInst{\CTerm{X}{s}{\vec{x}}{\varphi}}$ with $u < v$,
then $v \notin \AllInst{\CTerm{X}{s}{\vec{x}}{\varphi}}_\mathsf{v}$.
(1) is immediate from \Cref{lem:simple properties of value instantiation}~\Bfnum{3}.
To show~(2),
suppose
$u,v \in \AllInst{\CTerm{X}{s}{\vec{x}}{\varphi}}$ with $u < v$,
i.e.,
there exist  substitutions $\sigma,\theta$
such that 
(3) $u = s\sigma$,
(4) $\sigma(X) \subseteq \Val$,
(5) $\sigma \vDash \ECO{\vec{x}}{\varphi}$, 
(6) $v = u\theta$, and
(7) $v \not\leqslant u$.
Note that from (3)--(5) 
we obtain $u \in \AllInst{\CTerm{X}{s}{\vec{x}}{\varphi}}$,
and thus $v \in \AllInst{\CTerm{X}{s}{\vec{x}}{\varphi}}$ 
follows by \Cref{lem:immediate properties of interpretation}~\Bfnum{2} and~\Bfnum{6}.
Now we suppose $v \in \AllInst{\CTerm{X}{s}{\vec{x}}{\varphi}}_\mathsf{v}$
and derive a contradiction.
Let $s = s[x_1,\ldots,x_m,y_1,\ldots,y_n]_{p_1,\ldots,p_m,q_1,\ldots,q_n}$
with $\Pos_X(s) = \SET{ p_1,\ldots,p_m }$
and $\Pos_{\xV \setminus X}(s)  = \SET{ q_1,\ldots,q_n }$.
Then, by~(3) and~(4), 
$u = s[v_1,\ldots,v_m,y_1\sigma,\ldots,y_n\sigma]$
with $v_1,\ldots,v_m \in \Val$.
By~(6), 
$v = u\theta = s[v_1,\ldots,v_m,y_1\sigma\theta,\ldots,y_n\sigma\theta]$.
From our assumption 
that $v \in \AllInst{\CTerm{X}{s}{\vec{x}}{\varphi}}_\mathsf{v}$,
it follows that the sequence of terms
$\langle y_1\sigma\theta,\ldots,y_n\sigma\theta \rangle$
is a sequence of variables which is a renaming of $\langle y_1,\ldots,y_n \rangle$.
Hence, $v \leqslant u$.
This is a contradiction to~(7).
\end{proof}


\LemmaCharacterizationOfWeakRewriteStepsByValueInterpretation*
\begin{proof}
\begin{enumerate}
\item 
The specific substitution $\hat\delta$ in the proof of
\Cref{lem:Characterization of Weak Rewrite Steps}
satisfies $\hat\delta(x) = x$ for $x \in \Var(s) \setminus X$,
as $(\Var(\pi\gamma) \cup \ExVar(\rho) \setminus X) \cap \Var(s) =  \varnothing$.
Thus, it follows that 
$u = s\hat\delta \in \AllInst{\CTerm{X}{s}{\vec{x}}{\varphi}}_\mathsf{v}$.

\item 
In the proof of \Cref{lem:Characterization of Weak Rewrite Steps},
assume moreover that 
(10') $\delta(\Var(t) \setminus Y) \subseteq \xV$
and 
(10'') for any $x,y \in \Var(t) \setminus Y$,
$x \neq y$ implies $\delta(x) \neq \delta(y)$.
From the definitions of $\delta'$ and $\delta''$,
clearly they also satisfy these properties.
Since $(\Var(t) \setminus Y) \subseteq (\Var(s) \setminus X)$,
it follows that 
$u = s\delta'' \in \AllInst{\CTerm{X}{s}{\vec{x}}{\varphi}}_\mathsf{v}$.

\item
In the proof of \Cref{lem:Characterization of Weak Rewrite Steps},
assume moreover that 
$\delta(\Var(s) \setminus X) \subseteq \xV$
and 
for any $x,y \in \Var(s) \setminus X$,
$x \neq y$ implies $\delta(x) \neq \delta(y)$.
To show $t\delta' \in \AllInst{\CTerm{Y}{t}{\vec{y}}{\psi}}_\mathsf{v}$,
it remains to show
(1) $\delta'(\Var(t) \setminus Y) \subseteq \xV$
and 
(2) for any $x,y \in \Var(t) \setminus Y$,
$x \neq y$ implies $\delta(x) \neq \delta(y)$.
They follow as $\ExVar(\rho) \subseteq Y$,
$\Var(t)\subseteq \Var(s) \cup \ExVar(\rho)$,
and 
$(\Var(t)\setminus Y) \subseteq (\Var(s) \setminus X)$.

\end{enumerate}
\end{proof}

\LemmaNormalFormsOfWeakReduction*
\begin{proof}
Let $u \in \AllInst{\CTerm{X}{s}{\vec{x}}{\varphi}}_\mathsf{v}$.
Then, $u = s\mu$ for a substitution $\mu$ such that 
(1) $\mu(X) \subseteq \Val$,
(2) $\mu(\Var(s) \setminus X) \subseteq \xV$, and
(3) $\mu \vDash \ECO{\vec{x}}{\varphi}$.
Let $\rho: \CRu{Z}{\ell}{r}{\pi}$ be a left-linear constrained rewrite rule
without loss of generality such that $\Var(\rho) \cap \Var(s,\varphi) = \varnothing$,
and 
suppose that $u \to^p_{\AllInst{\rho}} v$.
More precisely,
suppose that $\sigma$ is a substitution such that
(4) $\Dom(\sigma) = \VDom(\sigma) = Z \cap \Var(\ell,r)$
and 
(5) $\sigma \vDash_{\xM} \ECO{\pvec{z}'}{\pi}$ where
$\SET{\pvec{z}'} = \Var(\pi) \setminus \Var(\ell,r)$,
(6) $u|_p = \ell\sigma\delta$, and 
(7) $v = u[r\sigma\delta]_p$
so that $\ell\sigma \to r\sigma \in \AllInst{\rho}$
and $u \to^p_{\ell\sigma \to r\sigma} v$.

We first show
$\CTerm{X}{s}{\vec{x}}{\varphi} \leadsto^p_{\rho,\gamma} \CTerm{Y}{t}{\vec{y}}{\psi}$
for some $\gamma$.
From~(1) and~(2), we have $\mu(\Var(s)) \subseteq \Val \cup \xV$.
Thus, $\Pos(s|_p) = \Pos(s\mu|_p) =  \Pos(\ell\sigma)$.
Hence, as $\ell$ is linear and value-free,
there exists a substitution $\gamma$ 
such that $\ell\gamma = s|_p$ and $\Dom(\gamma) = \Var(\ell)$.
Using~(6), we now have $\ell\gamma\mu = s|_p \mu = s\mu|_p =  u|_p = \ell\sigma\delta$,
i.e., 
(8) $\mu(\gamma(x)) = \delta(\sigma(x))$ for all $x \in \Var(\ell)$.

Let $x \in \Var(\ell) \cap Z$.
Then (9) $\mu(\gamma(x)) = \delta(\sigma(x)) = \sigma (x) \in \Val$ by~(4) and~(8).
Thus $\gamma(x) \in \Val \cup \xV$.
Moreover, $\gamma(x) \in \Val \cup \Var(s)$ by $\ell\gamma = s|_p$.
Hence $\gamma(x) \in \Val \cup X$;
for, if $\gamma(x) \in \Var(s) \setminus X$
then $\mu(\gamma(x)) \in \xV$ by~(2),
which contradicts $\mu(\gamma(x)) \in \Val$.
Thus, (10) $\gamma(x) \in \Val \cup X$ for all $x \in \Var(\ell) \cap Z$.

Let $\SET{\vec{z}} = \Var(\pi) \setminus \Var(\ell)$.
From (5) and $\SET{\pvec{z}'} \subseteq \SET{\vec{z}}$,
we have $\vDash_\sigma \ECO{\vec{z}}{\pi}$.
As $\FVar(\ECO{\vec{z}}{\pi}) = \Var(\ell) \cap Z$,
we have $\vDash_{\mu \circ \gamma} \ECO{\vec{z}}{\pi}$ by~(9).
Thus, $\vDash_{\mu} (\ECO{\vec{z}}{\pi})\gamma$.
By $\SET{\vec{z}} \cap \Dom(\gamma) = 
\SET{\vec{z}} \cap \Ran(\gamma) = \varnothing$,
we have $\vDash_{\mu} \ECO{\vec{z}}{\pi\gamma}$.
Thus, together with~(3),
we have 
(11) $\vDash_\mu (\ECO{\vec{x}}{\varphi}) \land (\ECO{\vec{z}}{\pi\gamma})$.

All in all, we know that
$\Dom(\gamma) = \Var(\ell)$,
$s|_p = \ell\gamma$,
$\gamma(x) \in \Val \cup X$ for all $x \in \Var(\ell) \cap Z$ by~(10), and
$(\ECO{\vec{x}}{\varphi}) \land (\ECO{\vec{z}}{\pi\gamma})$ is satisfiable by~(11).
Thus, we have $\CTerm{X}{s}{\vec{x}}{\varphi} \leadsto^p_{\rho,\gamma} \CTerm{Y}{t}{\vec{y}}{\psi}$,
where 
(15) $t = s[r\gamma]_p$,
(16) $\psi = \varphi \land \pi\gamma$,
(17) $\SET{\vec{y}} = \Var(\psi) \setminus \Var(t)$, and
(18) $Y = \ExVar(\rho) \cup (X \cap \Var(t))$.
It remains to show
$v \in \AllInst{\CTerm{Y}{t}{\vec{y}}{\psi}}_\mathsf{v}$.

Let $\mu'$ by a substitution given 
by $\mu'(x) = \mu(x)$ for $x \in \Var(s) \cup X$,
$\mu'(x) = \sigma(x)$ for each $x \in \ExVar(\rho)$,
and $\mu'(x) = x$ otherwise.
Clearly, (19) $s\mu = s\mu'$.

Let $y \in \Var(r)$.
If $y \in \ExVar(\rho)$ then 
$\mu'(\gamma(y)) = \mu'(y) = \sigma(y) = \delta(\sigma(y))$,
as $\sigma(y) \in \Val$ by~(4).
Otherwise $y \in \Var(\ell)$.
Then, $\mu(\gamma(y)) =  \delta(\sigma(y))$ by~(8).
Also, as $\gamma(y)$ is a subterm of $s$,
$\mu'(\gamma(y)) = \mu(\gamma(y))$ by~(19).
Hence $\mu'(\gamma(y)) =  \delta(\sigma(y))$.
Thus, $\mu'(\gamma(y)) =  \delta(\sigma(y))$ for all $y \in \Var(r)$,
and so (20) $r\gamma\mu' =  r\sigma\delta$.

Now, using~(19) and~(20), we have 
(21)
$v = u[r\sigma\delta]_p 
= s\mu[r\sigma\delta]_p 
= s\mu'[r\gamma\mu']_p 
= s[r\gamma]_p\mu' 
= t\mu'$.

Let $y \in Y = \ExVar(\rho) \cup (X \cap \Var(t))$.
If $y \in \ExVar(\rho)$,
then $y \in \ExVar(r) \cap Z$,
and hence, $\mu'(y) = \sigma(y) \in \Val$ by~(4).
If $y \in X \cap \Var(t)$
then $\mu'(y) = \mu(y)$ by $y \in X$,
and hence it follow from $y \in X$ that $\mu'(y) = \mu(y) \in \Val$ by~(1).
Thus, (22) $\mu'(Y) \subseteq \Val$ holds.

Let $y \in \Var(t) \setminus Y$.
Then $y \in \Var(t) \setminus (\ExVar(\rho) \cup (X \cap \Var(t)))$.
Thus, $y \in \Var(t) \setminus (\ExVar(\rho) \cup X)$.
Clearly, if $y \notin \Var(s) \cup X$, then $\mu'(x) = x \in \xV$.
If $y \in \Var(s) \cup X$, 
then $y \in \Var(s) \setminus X$,
and hence $\mu'(y) = \mu(y) \in \xV$ by (2).
Thus, (23) $\mu'(\Var(t) \setminus Y) \subseteq \xV$ holds.

We now show $\vDash_{\mu'} \ECO{\vec{y}}{\psi}$.
From~(11), we have
$\vDash_\mu \ECO{\vec{x}}{\varphi}$
and $\vDash_\mu \ECO{\vec{z}}{\pi\gamma}$.
As $\FVar(\ECO{\vec{x}}{\varphi}) \subseteq X$,
we have $\vDash_{\mu'} \ECO{\vec{x}}{\varphi}$.
Also,
by $\FVar(\ECO{\vec{z}}{\pi\gamma}) \subseteq \Var(s)$
we have $\vDash_{\mu'} \ECO{\vec{z}}{\pi\gamma}$.
Since $\SET{\vec{x}} \cap \SET{\vec{z}} = \varnothing$,
one can extend $\mu'$ to $\mu''$ 
so that $\vDash_{\mu''} \varphi \land \pi\gamma$.
Since $\SET{\vec{x},\vec{z}} \subseteq \SET{\vec{y}}$,
it follows (24) $\vDash_{\mu'} \ECO{\vec{y}}{\psi}$.

Thus, from~(21)--(24),
we conclude
$v \in \AllInst{\CTerm{Y}{t}{\vec{y}}{\psi}}_\mathsf{v}$.
\end{proof}

\TheoremNormalFormsOfWeakReduction*
\begin{proof}
($\Longleftarrow$)
We show the contraposition.
Suppose $\CTerm{X}{s}{\vec{x}}{\varphi} \leadsto^p_\rho \CTerm{Y}{t}{\vec{y}}{\psi}$.
From 
\Cref{lem:Characterization of Weak Rewrite Steps by Value Interpretation}~\Bfnum{1},
there exists 
$u \in \AllInst{\CTerm{X}{s}{\vec{x}}{\varphi}}_{\mathsf{v}}$ and 
$v$ such that 
$u \to^p_{\AllInst{\rho}} v$.

($\Longrightarrow$)
We proceed with a proof by contradiction.
Assume there exists $u \in \AllInst{\CTerm{X}{s}{\vec{x}}{\varphi}}_\mathsf{v}$
such that $u \to_{\AllInst{\xR}}^p v$ for some $v$.
Then, by the definition of value instantiation,
$u = s\mu$ for some $X$-valued substitution $\mu$
such that $\Dom(\mu) = X$ and $\mu \vDash_\xM \ECO{\vec{x}}{\varphi}$.
By $u \to_{\AllInst{\xR}}^p v$,
we have $\langle u|_p, v|_p \rangle \in \AllInst{\rho}$ for some $\rho \in \xR$.
However, it follows from \Cref{lem:Normal Forms of Weak Reduction}
that $\CTerm{X}{s}{\vec{x}}{\varphi} \leadsto^p_\rho \CTerm{Y}{t}{\vec{y}}{\psi}$.
\end{proof}

\LemmaInstantiationsAreRedexIsEnoughForRewriteSteps*
\begin{proof}
Let $\rho: \CRu{Z}{\ell}{r}{\pi}$ be a left-linear and left-value-free constrained rewrite rule
satisfying $\Var(\rho) \cap \Var(s,\varphi) = \varnothing$.
Suppose that for any $u \in \AllInst{\CTerm{X}{s}{\vec{x}}{\varphi}}_\mathsf{v}$
there exists $v$ such that
$\langle u|_p,v|_p  \rangle \in 
\SET{ \langle \ell\sigma,r \sigma \rangle \mid
Z \cap \Var(\ell,r)  \subseteq \VDom(\sigma),
\sigma \vDash_\xM \ECO{\pvec{z}'}{\pi}
}$,
where $\SET{\pvec{z}'} = \Var(\pi) \setminus \Var(\ell,r)$.

It suffices to show 
that $\CTerm{X}{s}{\vec{x}}{\varphi}$ has a $\rho$-redex at $p$
by some $\gamma$.
That is, 
we have to show there exists a substitution $\gamma$ such that
(i) $\Dom(\gamma) = \Var(\ell)$,
(ii) $s|_p = \ell\gamma$,
(iii) $\gamma(x) \in \Val \cup X$ for all $x \in \Var(\ell) \cap Z$, and
(iv) $\vDash_\xM (\ECO{\vec{x}}{\varphi}) \Rightarrow (\ECO{\vec{z}}{\pi\gamma})$,
where $\SET{\vec{z}} = \Var(\pi) \setminus \Var(\ell)$.

First, we show there exists $s|_p$ matches with $\ell$,
that is, $s|_p = \ell\gamma$ for some substitution $\gamma$.

By the satisfiability of $\CTerm{X}{s}{\vec{x}}{\varphi}$,
there exists $\delta$ such that 
$\delta \vDash \ECO{\vec{x}}{\varphi}$.
Then $u = s\delta \in \AllInst{\CTerm{X}{s}{\vec{x}}{\varphi}}_\mathsf{v}$.
Note here that 
(0) $\Pos(u) = \Pos(s\delta) = \Pos(s)$
and $\delta(x) \in \Val$ for all $x \in X$
and $\delta(y) \in \xV$ for all $y \in \xV \setminus X$
by the definition of value instantiation.
By our assumption,
there exists a substitution $\sigma$
such that (1) $u|_p = \ell\sigma$,
(2) $Z \cap \Var(\ell,r)  \subseteq \VDom(\sigma)$,
and (3) $\sigma \vDash_\xM \ECO{\pvec{z}'}{\pi}$
where $\SET{\pvec{z}'} = \Var(\pi) \setminus \Var(\ell,r)$.

Let $q \in \FPos(\ell)$.
Then $\ell_q = \ell\sigma|_q = u|_{p.q} = s\delta_{p.q}$.
If $p.q \in \FPos(s)$, then $\ell(q) = s\delta(p.q) = s(p.q)$.
Otherwise $p.q \in \VPos(s)$, as $\Pos(s) = \Pos(s\delta)$.
Since $q \in \FPos(\ell)$, it follows 
from $\ell(q) = s\delta(p.q)$  that $s\delta(p.q) \in \Val$.
However, this contradicts our assumption that $\ell$ is value-free.
Thus, we have shown $s(p.q) = \ell(q)$ for all $q \in \FPos(\ell)$.
Hence, there exists a substitution $\gamma$
such that $s|_p = \ell\gamma$ by the left-linearity of $\rho$.
without loss of generality one can assume $\Dom(\gamma) = \Var(\ell)$.
Thus, we have shown~(i) and~(ii).

Now we have $\ell\gamma\delta = (s|_p)\delta = s\delta|_p = u|_p = \ell\sigma$.
Let $z \in \Var(\ell) \cap Z$ and $q$ be a position such that $\ell(q) = z$.
Then we have $\sigma(z) \in \Val$ by~(2).
Thus,
$\delta(\gamma(z)) = \ell\gamma\delta_q = \ell\sigma|_q = \sigma(z) \in \Val$.
Then $\gamma(z) \in \Val \cup \xV$.
If $\gamma(z) = x \in \xV$, then $\delta(x) \in \Val$,
and hence $x \in X$ by~(0).
Thus, $\gamma(z) \in \Val \cup X$.
Therefore, it holds that
$\gamma(z) \in \Val \cup X$ for any $z \in \Var(\ell) \cap Z$.
Thus, we have shown~(iii).

It remains to show~(iv).
For this, suppose $\xi$ is a valuation such that $\vDash_{\xM,\xi} \ECO{\vec{x}}{\varphi}$.
Take a substitution $\hat\xi = \xi|_X$.
Then $\hat\xi$ is $X$-valued
and $\hat\xi \vDash \ECO{\vec{x}}{\varphi}$
by $\FVar(\ECO{\vec{x}}{\varphi}) \subseteq X$.
Also, as $\SET{ \vec{x} } \cap X = \varnothing$,
we have $\inter{\ECO{\vec{x}}{\varphi\hat\xi}}_{\xM} = \mathsf{true}$.
Let $u' = s\hat \xi$.
Then, $u' \in \AllInst{\CTerm{X}{s}{\vec{x}}{\varphi}}_\mathsf{v}$.
Thus, by our assumption,
there exists a substitution $\sigma'$
such that (1') $u'|_p = \ell\sigma'$,
(2') $Z \cap \Var(\ell,r)  \subseteq \VDom(\sigma')$,
and (3') $\sigma' \vDash_\xM \ECO{\pvec{z}'}{\pi}$
where $\SET{\pvec{z}'} = \Var(\pi) \setminus \Var(\ell,r)$.

Let $x \in \FVar(\ECO{\vec{z}}{\pi\gamma})$,
i.e., $x \in \Var(\pi\gamma) \setminus (\Var(\pi) \setminus \Var(\ell))$.
Then $x \in \SET{ \gamma(z) \in \xV \mid z \in \Var(\ell) \cap Z }$.
Thus, let $q,z$ be such that 
$z \in \Var(\ell) \cap Z$, $x = \gamma(z)$,
and $\ell|_q = z$.
Then $\sigma'(z) \in \Val$ by~(2'),
and hence using~(1'), we have
$u'|_q = \ell\sigma'|_q  = (\ell|_q)\sigma' = \sigma'(z) \in \Val$.
Then, as $s \hat\xi = u' \in \AllInst{\CTerm{X}{s}{\vec{x}}{\varphi}}_\mathsf{v}$,
we have $s|_q \in X \cup \Val$.
Then,
$\hat\xi(x) 
= \hat\xi(\gamma(z)) 
= \ell|_q\gamma\hat\xi
= (\ell\gamma|_q)\hat\xi
= (s|_q)\hat\xi \in \Val$, as $\hat\xi = \xi|_X$.
Thus, for any $x \in \FVar(\ECO{\vec{z}}{\pi\gamma})$,
we have $\hat\xi(x) \in \Val$.
Furthermore, 
$\hat\xi(\gamma(z)) = \hat\xi(x) = (s|_q)\hat \xi = u'|_q =  \sigma'(z)$.
Thus, 
$\vDash_{\xM,\hat\xi} \ECO{\vec{z}}{\pi\gamma}$ by~(3').
Hence,
$\vDash_{\xM,\xi} \ECO{\vec{z}}{\pi\gamma}$.
Thus, it follows that
$\vDash_{\xM,\xi} (\ECO{\vec{x}}{\varphi}) \Rightarrow (\ECO{\vec{z}}{\pi\gamma})$
for any valuation $\xi$, and therefore~(iv) holds.
\end{proof}

\PropositionIntantiationNormalImpliesNormalwrtPartialRewriting*
\begin{proof}
We first note that for any non-constrained term $u$,
$u \to_{\rho} v$ if and only if $u \to_{\AllInst{\rho}} v$.
We show a contraposition.
Suppose 
$\CTerm{X}{s}{\vec{x}}{\varphi} \leadsto_{\xR} \CTerm{Y}{t}{\vec{y}}{\psi}$.
Then, by \Cref{lem:Characterization of Weak Rewrite Steps},
there exists $u \in \AllInst{\CTerm{X}{s}{\vec{x}}{\varphi}}$
such that $u \to_{\AllInst{\xR}} v$ for some $v$.
Thus, $u \to_{\xR} v$ for some $v$.
\end{proof}

\end{document}